\documentclass[11pt]{article}
\usepackage{fullpage}
\usepackage[top=0.9in,bottom=1.1in,left=1.1in,right=1.1in]{geometry}
\usepackage{graphicx}
\usepackage{amsmath,amssymb,graphicx,paralist,amsthm}
\RequirePackage{hyperref}

\usepackage{times}
\usepackage{bm}
\usepackage{natbib}
\usepackage{anyfontsize}
\usepackage{subcaption}
\usepackage{authblk}

\newtheorem{theorem}{Theorem}
\newtheorem{lemma}[theorem]{Lemma}
\newtheorem{corollary}[theorem]{Corollary}
\newtheorem{proposition}{Proposition}
\newtheorem{condition}{Condition}

\newcommand{\beq}{\begin{equation}}
\newcommand{\eeq}{\end{equation}}
\newcommand{\beas}{\begin{align*}}
\newcommand{\eeas}{\end{align*}}
\newcommand{\bea}{\begin{align}}
\newcommand{\eea}{\end{align}}
\newcommand{\bei}{\begin{itemize}}
    \newcommand{\eei}{\end{itemize}}
\newcommand{\ben}{\begin{enumerate}}
    \newcommand{\een}{\end{enumerate}}
\newcommand{\bet}{\begin{theorem}}
    \newcommand{\eet}{\end{theorem}}
\newcommand{\bel}{\begin{lemma}}
    \newcommand{\eel}{\end{lemma}}
\newcommand{\bep}{\begin{proposition}}
    \newcommand{\eep}{\end{proposition}}
\newcommand{\bed}{\begin{definition}}
    \newcommand{\eed}{\end{definition}}
\newcommand{\bec}{\begin{corollary}}
    \newcommand{\eec}{\end{corollary}}
\newcommand{\bex}{\begin{example}}
    \newcommand{\eex}{\end{example}}

\newcommand{\cov}{\text{cov}}
\newcommand{\tr}{\text{tr}}
\newcommand{\var}{\text{var}}

\newcommand{\bzero}{\mathbf{0}}
\newcommand{\bone}{\mathbf{1}}
\newcommand{\bI}{\mathbf{I}}

\newcommand{\cJ}{\mathcal{J}}
\newcommand{\cN}{\mathcal{N}}
\newcommand{\ELR}{\textsc{elr}}

\newcommand{\ignore}[1]{}

\newenvironment{psmallmatrix}
  {\left(\begin{smallmatrix}}
  {\end{smallmatrix}\right)}

\title{Empirical Likelihood Inference of Variance Components in Linear Mixed-Effects Models}

\author[1]{J. Zhang}
\author[2]{W. Guo} 
\author[3]{J.S.\ Carpenter} 
\author[2,4]{Andrew Leroux} 
\author[2]{K.R.\ Merikangas} 
\author[5]{N.G.\ Martin}
\author[3]{I.B.\ Hickie}
\author[1]{H. Shou}
\author[1]{H. Li \footnote{hongzhe@pennmedicine.upenn.edu}}

\affil[1]{Department of Biostatistics, Epidemiology and Informatics, Perelman School of Medicine, University of Pennsylvania, Philadelphia, Pennsylvania 19104, U.S.A.}
\affil[2]{Genetic Epidemiology Research Branch, National Institute of Mental Health, National Institutes of Health}
\affil[3]{The University of Sydney's Brain and Mind Centre}
\affil[4]{Department of Biostatistics \& Informatics, University of Colorado}
\affil[5]{QIMR Berghofer Medical Research Institute}

\date{}                                       
\begin{document}

\maketitle

\begin{abstract}
Linear mixed-effects models are widely used in analyzing repeated measures data, including clustered and longitudinal data, where inferences of both fixed effects and variance components are of importance.  
Unlike the fixed effect inference that  has been well studied, inference on the variance components is more challenging due to null value being on the boundary and  the nuisance parameters of the fixed effects. Existing methods often require strong distributional assumptions on the random effects and random errors.  
In this paper, we develop empirical likelihood-based methods for the inference of the variance components in the presence of  fixed effects. 
A nonparametric version of the Wilks' theorem for the proposed empirical likelihood ratio  statistics for variance components  is derived. We also develop an empirical likelihood test for multiple variance components related to a sequence of correlated outcomes.  Simulation studies demonstrate that the proposed methods exhibit better type 1 error control than the commonly used likelihood ratio tests when the Gaussian distributional assumptions of the random effects are violated. We apply the methods to investigate the heritability of physical activity as measured by wearable device in the Australian Twin study and observe that such activity is heritable only in  the quantile range from 0.375 to 0.514.

\end{abstract}

\section{Introduction}
Longitudinal and clustered data commonly arise from observational  studies or clinical trials, where subjects are measured repeatedly over time or within a cluster. The repeated measures within a subject or a cluster are often correlated. To analyze such data, linear mixed-effects models that incorporate both fixed and random effects are widely used.
Many statistical methods have been developed for such linear mixed-effects models, especially methods for inference of the fixed effects. However, inference on  the variance components is less studied and often requires strong distributional assumptions on the random effects and the error terms. When the underlying distributions are known,  classical inference  methods, including  the likelihood ratio tests and the score tests, can be applied. However, these parametric methods are often restrictive and not robust  if the model assumptions are violated. 

Empirical likelihood (\textsc{el}) method, as an alternative to parametric likelihood-based methods, was first proposed by \cite{owen1988empirical} and has been applied to many statistical inference problems.  Without a prespecified distributional assumption  on the data,  \textsc{el}  methods incorporate side information through constraints or prior distributions and have favorable statistical  properties, including but not limited to Bartlett correctability, transformation invariance, better coverage accuracy of the corresponding confidence internals and greater power. A comprehensive introduction to  \textsc{el}  methods  can be found in \cite{owen2001empirical}.
 \textsc{el}  methods have been applied to inferences of mixture models \citep{zou2002empirical} and censored survival data \citep{chang2016empirical}, and  have also been considered for longitudinal data modeling.  For example, \cite{you2006block} proposed a block  \textsc{el}  method for inference of the regression parameters assuming a working independence covariance,  and  \cite{xue2007empirical} considered a semiparametric regression model, where the repeated within-subject measures are summarized as a function over time in order to address the dependence issue.  \cite{wang2010generalized} proposed a generalized  \textsc{el}  method that takes into account the within-subject correlations. 
\cite{li2013empirical} defined an empirical likelihood ratio (\ELR) test by utilizing the extended score from quadratic inference functions for longitudinal data, which does not involve direct estimation of the correlation parameters.

The  \textsc{el}  methods  mentioned above only focus on the inference of fixed effects in linear mixed effect models.
In this paper, we consider a general setting of linear mixed-effects models and develop   \textsc{el}  methods for the inference of the variance components.
Specifically, suppose there are $n$ subjects and denote by $n_i$ the number of repeated measures for the $i$th subject. For the $i$th subject, we observe a response vector $y_i\in R^{n_i}$, an $n_i\times p$ design matrix $X_i$ for the fixed effects $\beta^*\in R^p$, and $d$ $n_i\times n_i$ semi-positive design matrices $\Phi_{iq}~(q=1,\cdots,d)$ for the variance components $\theta^*\in (R_+\cup \{0\})^d$.
The general  linear mixed-effects model can be written as
\begin{equation}\label{eq:yi}
    y_i=X_i\beta^*+r_i, \quad i=1,\cdots,n,
\end{equation}
where $ r_i\in R^{n_i}$ is a zero-mean random variable with variance-covariance $H_i(\theta^*)$. 
We assume that $H_i(\theta^*)$ has a linear structure, 
\[
H_i(\theta^*) = \sum_{q=1}^d\theta^*_q\Phi_{iq}, \quad \theta^*=(\theta^*_1,\cdots,\theta^*_d)^T= (\theta^*_1,\theta^{*T}_{(1)})^T,
\]
where $\theta^*=(\theta^*_1,\cdots,\theta^*_d)^T$ is the vector of the variance components. 
We emphasize that this general setting does not require any assumptions on  the distributions of the data or the distributions of the random effects.

In many real applications, we are interested in making statistical inference on the variance components $\theta^*$ in model \eqref{eq:yi}.  For example, in the study of heritability based on twin data,  each monozygotic twin  or dizygotic twin   is treated as one cluster, and the linear variance structure can be constructed based on the twin type (see details in Section \ref{sec:realapp}). 
In the heritability analysis, a key question is whether there exists an genetic effect, which motivates us to study the inference of one of the variance components, say, $\theta^*_1$. 
We propose to develop an  \textsc{el}  based inference method for $\theta^*_1$ without any assumptions on the random components. The method can effectively account for the nuisance parameters, including the unknown fixed effects $\beta^*$ and  the variance components $\theta^{*T}_{(1)}$. The key difficulty when compared to the  \textsc{el}  inference of the fixed effects is to deal with  the boundary value problem when $\theta^*_1=0$ in  local testing problem $H_0: \theta^*_1=\theta_1^0$. To solve the issues, we propose a new empirical likelihood ratio test by utilizing an unbiased estimator of $\beta^*$ under very mild conditions, and prove that the asymptotic distribution of the test statistic is a mixture of $\chi^2$ distribution. 

Motivated by heritability analysis of daily activity distribution  as measured by wearable device such as actigraphy, we also consider the setting when linear mixed-effects models are fitted to a sequence of  dependent outcomes.  The wearable device data have been increasingly collected for   continuous activity monitoring in large observational  or experimental studies \citep{burton2013activity,krane2014actigraphic}.  
In typical wearable activity tracking data, the activity is  measured at one-minute resolution over several days for a given subject. Such wearable device data with repeated measures enable us to account for day-to-day variability of the activity. 
Instead of focusing on the activity counts at any minute of the day, daily activity distribution  or the amount of time with  the activity count above a given threshold  provides a biologically meaningful measure of the activity traits.  
When the activity counts are summarized as distributions,  we can  consider the activity quantile profile as a phenotype measure. In analysis of such wearable device data, we fit a linear mixed-effects model for each of the activity level or quantile  $y_i(t)$ at $t$.  
Denote by $\theta^*(t)$ the variance components for activity profile at level $t$. We are then interested in testing the global null $H_0:\theta^*_1(t) \equiv \theta_1^0,~t\in[t_1,t_2]$.
We develop  a max-type statistic for this global testing problem. Since the numerator of the proposed empirical likelihood ratio (\ELR) tests can be rewritten as the sum of approximately independent random variables over different subjects, a random perturbation method is developed to  approximate the $p$-values of the proposed global test.

We first introduce some notation. Denote by $(A)_{-1}$ the submatrix of $A$ without the first column of $A$. 
For two vectors or matrices $A$ and $B$ of compatible dimension, define the inner product $\langle A,B\rangle=\tr(A^TB)$.
For a matrix $D_{m\times n}=(D_1,\cdots,D_n)$, where $D_i$ is the $i$th column of $D$, the vectorized $D$ is defined by $(D_1^T,\cdots,D_n^T)^T$. 
Let $E(x)$ and $\var(x)$ be the expectation and variance of a random vector $x$, and let $\cov(x,y)$ be the covariance of random vectors $x$ and $y$. When $x$ is a random matrix, $E(x)$ and $\var(x)$ represent the expectation and variance of the vectorized $x$. When $x$ and $y$ are random matrices, $\cov(x,y)$ denotes the covariance of the vectorized $x$ and vectorized $y$.
We use $a=O(b)$ to denote that $a$ and $b$ are of the same order, and $a=o(b)$ to denote that $a$ is of a smaller order than $b$.
We use  $x=O_p(y)$ to denote that $x$ and $y$ are of the same order in probability, and $x=o_p(y)$ to denote that $x$ is of a smaller order than y in probability.

\section{\ELR~test for the fixed effects $\beta^*$}\label{sec:coef}
Statistical tests for the fixed effects  in the linear mixed-effects model \eqref{eq:yi},  $H_0: \beta^*=\beta_0$,  has been well studied. We first  briefly review the subject-wise  \textsc{el}  method proposed in \cite{wang2010generalized}, where the covariance structure for each subject is considered. 

Let $\hat H_{in}$ be an estimator of $H_i$, and assume that $\hat H_{in}$ converges to some $H_i^*$ in probability uniformly over all $i = 1,\cdots,n$. One such a nonparametric sample covariance matrix $\hat H_{in}$ can be obtained using  a simple two-step procedure, including 
 estimating the residuals $\hat r_i=y_i-X_i\hat\beta$, where $\hat\beta$ is the least-squares estimator using working independence correlation matrices, and 
 solving the constrained optimization problem
$\min_{\theta\geq 0} \sum_{i=1}^n\|H_i(\theta)-\hat r_i\hat r_i^T\|_F^2$.
Let
\[
\phi_i(\beta) = X_i^T\hat H_{in}^{-1}(y_i-X_i\beta),
\]
which satisfies $E\{\phi_i(\beta)\}=0$ when $\beta$ is the true value. Denote by $p_i$ the point mass at the $i$th subject. The nonparametric empirical likelihood is defined as
\[
L_0(\beta) = \sup_{p_i}\Big\{\prod_{i=1}^np_i:p_i\geq 0,\sum_{i=1}^np_i=1,\sum_{i=1}^np_i\phi_i(\beta)=0\Big\}.
\]
Since it can be proved that $\max_\beta L_0(\beta)=1/n^n$ \citep{owen2001empirical}, \cite{wang2010generalized} proposed the \ELR~statistic
\[
\ELR_0(\beta_0) = \frac{L_0(\beta_0)}{\max_\beta L_0(\beta)}=n^n L_0(\beta_0).
\]

To obtain the asymptotic distribution of the \ELR~statistic, the following  regularity conditions are needed. 
\begin{condition}\label{cond:beta-1}
    As $n\rightarrow\infty$, $P(0\in ch\{\phi_1(\beta_0),\cdots,\phi_n(\beta_0)\})\rightarrow 1$, where $ch\{\}$ is the convex hull.
\end{condition}
\begin{condition}\label{cond:beta-2}
    The limit $\lim_{n\rightarrow\infty}n^{-1}\sum_{i=1}^nX_i^TH_i^{*-1}H_iH_i^{*-1}X_i$ exists and is positive definite.
\end{condition}
\begin{condition}\label{cond:beta-3}
    The expectation $E\|\phi_i(\beta_0)\|_2^{2+\gamma_1}$ are bounded uniformly for some $\gamma_1> 0$.
\end{condition}
\begin{condition}\label{cond:beta-4}
    Let $\hat G_{in}=\hat{H}_{in}^{-1}$ with element $\hat g_{ijk}$, $x_{ij}^T$ be the $j$th row of $X_i$, and $r_{ik}$ be the $k$th element of $r_i$.
    For each pair $i$ and $i'$ with $i, i'=1,\cdots,n$ and $i\neq i'$, $\hat g_{ijk}-\hat g_{-(i,i')jk}=O_p(n^{-1})$ and sufficient moment conditions are satisfied so that $E(\hat B_{ii'})=O(n^{-1})$ and $E(\hat B_{ii'}\hat B_{ii'}^T)=O(n^{-2})$, where $\hat g_{-(i,i')jk}$ is $\hat g_{ijk}$ but computed with all the data except for subjects $i$ and $i'$ and $\hat B_{ii'}=\sum_{j=1}^{n_i}\sum_{k=1}^{n_i}(\hat g_{ijk}-\hat g_{-(i,i')jk})x_{ij}r_{ik}$.  
\end{condition}

Conditions \ref{cond:beta-1}--\ref{cond:beta-3} are common conditions for the empirical likelihood methods \citep{owen1991empirical}. 
Condition \ref{cond:beta-4} assumes  mild constraints on $\hat H_{in}^{-1}$ to ensure that the difference between the statistic $\ELR_0(\beta_0)$ defined with $\hat H_{in}$ and the one using $H_i^*$ vanishes as $n\rightarrow\infty$.
Under these regularity conditions, the following theorem provides the asymptotic distribution of the \ELR~test $\ELR_0(\beta_0)$   \citep{wang2010generalized} under the null.
\begin{theorem}\label{th:beta1}
    Under the regularity conditions \eqref{cond:beta-1}--\eqref{cond:beta-4}, as $n\rightarrow\infty$, $-2\log \ELR_0(\beta_0) \rightarrow \chi_p^2$ in distribution under the null hypothesis $H_0: \beta^*=\beta_0$.
\end{theorem}

The asymptotic result only requires that the $\hat H_{in}$ converge uniformly to some $H_i^*$, which may not be the true $H_i$ \citep{wang2010generalized}.
When the correlation structure is correctly specified, the estimator $\hat H_{in}$ is a consistent estimator of $H_i^*=H_i$. 
The statistic defined with the true $H_i$ is asymptotically locally most powerful among all the choices of the weight matrices.

\section{\ELR~test for the variance component $\theta_1^*$ }\label{sec:local2}
We  consider the local test $H_0:\theta^*_1=\theta_1^0$ in the framework of the empirical likelihood, including the null  $H_0: \theta_1^*=0$, which is of the most interest. 
We define $r_i=y_i-X_i\beta^*$ and $R_i=r_ir_i^T$. Since $E(r_i)=0$ and $\var(r_i)=H_i(\theta^*)$, we have 
\[
R_i = H_i(\theta^*) + \delta_i = \sum_{q=1}^d\theta^*_q\Phi_{iq}+\delta_i,
\]
where $E({\delta_i})=0$ and $\var({\delta_i})$ exists.
Since $\beta^*$ is unknown, we  first need an estimator of $\beta^*$, denoted by $\hat\beta$. One simple choice is  the least-squares estimator using the all data. Specifically, we stack the data from all subjects by denoting $X=(X_1^T,\cdots,X_n^T)^T$, $y=(y_1^T,\cdots,y_n^T)^T$, and $r=(r_1^T,\cdots,r_n^T)^T$.  Model (\ref{eq:yi}) can be rewritten as
\[
y = X\beta^* + r.
\]
Then the least-squares estimator is $\hat\beta = (X^TX)^{-1}X^Ty$. For $i= 1,\cdots,n$, let 
$
\hat r_i = y_i-X_i\hat\beta = r_i + X_i(\beta^*-\hat\beta).
$
We have
\begin{align*}
    \hat R_i & =\hat r_i\hat r_i^T = r_ir_i^T + r_i(\beta^*-\hat\beta)^TX_i^T + X_i(\beta^*-\hat\beta)r_i^T+X_i(\beta^*-\hat\beta)(\beta^*-\hat\beta)^TX_i^T \\
    & = R_i + \hat\epsilon_i  = H_i(\theta^*) + \delta_i + \hat\epsilon_i,
\end{align*}
where $\hat\epsilon_i =  r_i(\beta^*-\hat\beta)^TX_i^T + X_i(\beta^*-\hat\beta)r_i^T+X_i(\beta^*-\hat\beta)(\beta^*-\hat\beta)^TX_i^T$.

To control the rates of $E({\hat\epsilon_i})$, $\cov({r_ir_i^T},{\hat\epsilon_j})$, and $\cov({\hat\epsilon_i},{\hat\epsilon_j})$, we need the following condition, which is also commonly used for empirical likelihood methods.
\begin{condition}\label{cond:1-2}
    The expectation $E\|r_i\|_2^{4+\gamma_1}$ are bounded uniformly for some $\gamma_1> 0$.
\end{condition}
Under Condition \ref{cond:1-2}, we see that the least-squares estimator $\hat\beta$ is good enough.
\begin{proposition}\label{prop:1}
    Assume that $n^{-1}X^TX\rightarrow\Sigma$ and $n^{-1/2}X^Tr\xrightarrow{d}\eta$ as $n\rightarrow\infty$,
    where $0<\|\Sigma\|_2,\|\Sigma^{-1}\|_2<\infty$, $E \eta=0$ and $E\|\eta\|_2^4=O(1)$.
    When Condition \ref{cond:1-2} holds and $\hat\beta = (X^TX)^{-1}X^Ty$, we have
    $E({\hat\epsilon_i})=O(n^{-1}),$ $i=1,\cdots,n$, and $\cov({r_ir_i^T},{\hat\epsilon_j}), ~\cov({\hat\epsilon_i},{\hat\epsilon_j})=O(n^{-2})$, $i,j=1,\cdots,n, i\neq j$.
\end{proposition}

Let $\Xi=(\Xi_{kl})_{d\times d}$ with $\Xi_{kl}=\sum_{i=1}^n\tr(\Phi_{ik}\Phi_{il})$, and  $\hat\Upsilon=(\hat\Upsilon_1,\cdots,\hat\Upsilon_d)^T$ with $\hat\Upsilon_k=\sum_{i=1}^n\tr(\Phi_{ik}\hat R_i)$.
We define
\[
\hat Z_i(\theta_1) = \tr\Bigg\{\Phi_{i1}\Bigg(\hat R_i - \Phi_{i1}\theta_1-\sum_{q=2}^d\hat{\theta}_q\Phi_{iq}\Bigg)\Bigg\},~i=1,\cdots,n,
\]
where 
\begin{equation}\label{eq:theta2}
    \hat\theta_{(1)} = (\hat\theta_2,\cdots,\hat\theta_q)^T = (\Xi^{-1})_{-1}^T\hat\Upsilon. 
\end{equation}
Since Proposition \ref{prop:1} implies $E \hat Z_{i}(\theta_1)=O(n^{-1})$ if $\theta_1$ is the true value (see \eqref{eq:hatZi} in the appendix),
we define the nonparametric likelihood as 
$$
L(\theta_1) = \max_{p_i}\Bigg\{\prod_{i=1}^{n}p_i |p_i\geq 0,\sum_{i=1}^{n}p_i=1,\sum_{i=1}^{n}p_i\hat Z_{i}(\theta_1) = 0\Bigg\}
$$
and the corresponding  \ELR~statistic as 
\begin{equation}\label{eq:elr-theta}
\ELR(\theta_1^0) = \frac{L(\theta_1^0)}{\max_{\theta_1\geq 0}L(\theta_1)}.
\end{equation}

If the true value $\theta_1^*=0$ (i.e., the null hypothesis under the case $\theta_1^0=0$), the denominator in (\ref{eq:elr-theta}) would not be $1/n^n$ as usual owing to the boundary value issue, and thus the existing results are inapplicable.
To derive the asymptotic distribution of the proposed test $\ELR(\theta_1^0)$, we assume the following condition similar to Condition \ref{cond:beta-1}.
\begin{condition}\label{cond:1-1}
    As $n\rightarrow\infty$, $P(0\in ch\{Z_1(\theta_1^0),\cdots,Z_n(\theta_1^0)\})\rightarrow 1$, where $Z_i(\theta_1^0)$ is defined as $\hat Z_i(\theta_1^0)$ with $\hat R_i$ replaced by $R_i$.
\end{condition}

Under Conditions \ref{cond:1-2} and \ref{cond:1-1}, we have the following theorem  on the asymptotic distribution of the \ELR~test under the null.
\begin{theorem}\label{th:3}
    Let $\hat c_{n}(\theta_1^0)=\hat\nu_{2n}^2(\theta_1^0)/\hat\nu_{1n}^2(\theta_1^0)$, where $\hat\nu_{1n}^2(\theta_1^0)$ is a consistent estimator of the asymptotic variance of $n^{-1/2}\sum_{i=1}^{n}\hat Z_{i}(\theta_1^0)$ and $\hat\nu_{2n}^2(\theta_1^0)=n^{-1}\sum_{i=1}^{n}\hat Z_{i}^2(\theta_1^0)$. 
    If $\theta^*_{(1)}\in R_+^{d-1}$, and Conditions  \ref{cond:1-2} and  \ref{cond:1-1}  hold, as $n\rightarrow\infty$, $\hat c_{n}(\theta_1^0)$ $\left\{-2\log \ELR(\theta_1^0)\right\}\rightarrow \chi_1^2$ in distribution when $\theta_1^0>0$, and $\hat{c}_{n}(0)\left\{-2\log \ELR(0)\right\}\rightarrow U_+^2$ in distribution, where $U\sim N(0,1)$ and $U_+=\max(U,0)$.
\end{theorem}

Although the \ELR~statistic $\hat c_n(\theta_1^0)\left(-2\log \ELR(\theta_1^0)\right)$ in Theorem \ref{th:3} involves optimizations in the numerator and denominator,
the following lemma shows that the statistic has  an asymptotically equivalent expression that can be used to calculate the statistic efficiently.  

\begin{lemma}\label{lem:stat1}
    If $\theta^*_{(1)}\in R_+^{d-1}$, then under Conditions \ref{cond:1-2} and \ref{cond:1-1}, 
\begin{align*}
     & \hat c_n(\theta_1^0)\left\{-2\log \ELR(\theta_1^0) \right\} \\
    = & 
    \begin{cases}
    \frac{\left\{n^{-1/2}\sum_{i=1}^n\hat Z_i(\theta_1^0)\right\}^2}{\hat\nu_{1n}^2(\theta_1^0)} +o_p(1), & \text{ if }\quad \theta_1^0>0, \\
   \frac{\left\{n^{-1/2}\sum_{i=1}^n\hat Z_i(\theta_1^0)\right\}^2}{\hat\nu_{1n}^2(\theta_1^0)}I(\sum_{i=1}^n\hat Z_i(\theta_1^0)\geq 0)+o_p(1), & \text{ if }\quad \theta_1^0=0.
    \end{cases}
\end{align*}    
\end{lemma}

We next provide an estimator of the asymptotic variance  of $n^{-1/2}\sum_{i=1}^n \hat Z_i(\theta_1^0)$.   We rewrite $\Xi$ as \[\Xi=\begin{psmallmatrix} E_{11} & E_{12} \\ E_{21} & E_{22}
\end{psmallmatrix}
\]
with $E_{11}$ being a scalar.
Let $F=E_{22}^{-1}E_{21}=(F_1,\cdots,F_{d-1})^T$ and $\alpha = 1 - E_{12}F/E_{11} \in (0,1]$. It can be verified that 
\begin{equation}\label{eq:zdm}
\sum_{i=1}^n\hat Z_i(\theta_1^0) = \sum_{i=1}^n\hat D_i(\theta_1^0) = \sum_{i=1}^n\hat M_i(\theta_1^0),
\end{equation}
where 
\begin{align*}
    \hat D_i(\theta_1^0)&=\alpha^{-1}\langle\Phi_{i1}-\sum_{q=1}^{d-1}F_q\Phi_{iq+1}, \hat R_i-\theta_1^0\Phi_{i1}\rangle, \\
    \hat M_i(\theta_1^0)&=\alpha^{-1}\langle\Phi_{i1}-\sum_{q=1}^{d-1}F_q\Phi_{iq+1}, \hat R_i-H_i((\theta_1^0,\hat\theta_{(1)}^T)^T)\rangle.
\end{align*}
In addition, for $i\neq j$,
\begin{align*}
    \cov({\hat R_{i}},{\hat R_{j}}) = & \cov({\delta_i}+{\hat\epsilon_i},{\delta_j}+{\hat\epsilon_j}) 
    =  \cov({\delta_i},{\hat\epsilon_j})+\cov({\hat\epsilon_i},{\delta_j})+\cov({\hat\epsilon_i},{\hat\epsilon_j})
    =  O(n^{-2})
\end{align*}
based on proposition \ref{prop:1}.
Therefore, ${\hat D_i(\theta_1^0)}~(i=1,\cdots,n)$ are asymptotically independent with expectation 
 $E(\hat D_i(\theta_1^0))=\alpha^{-1}\langle\Phi_{i1}-\sum_{q=1}^{d-1}F_q\Phi_{iq+1}, \sum_{q=2}^d\theta_q^*\Phi_{iq}\rangle+O(n^{-1})$, while
the expectations of $\hat Z_i(\theta_1^0)$ and $\hat M_i(\theta_1^0)$ are $O(n^{-1})$ (see \eqref{eq:hatZi} in the appendix). We have 
\begin{align*}\hat D_i(\theta_1^0)-E(\hat D_i(\theta_1^0))&=\alpha^{-1}\langle\Phi_{i1}-\sum_{q=1}^{d-1}F_q\Phi_{iq+1}, \hat R_i-H_i((\theta_1^0,(\theta_{(1)}^*)^T)^T)\rangle+O(n^{-1})\\\
  &=\hat M_i(\theta_1^0)+o_p(1).
  \end{align*}
Therefore,  
\begin{align}\label{eq:var}
    \var\big\{n^{-1/2}\sum_{i=1}^n\hat Z_i(\theta_1^0)\big\} = &  \frac{1}{n}\sum_{i=1}^n\{\hat D_i(\theta_1^0)-E(\hat D_i(\theta_1^0))\}^2+o_p(1) \notag \\
    =  & \frac{1}{n}\sum_{i=1}^n\hat M_i(\theta_1^0)^2 + o_p(1),
\end{align}
which leads a consistent estimator of the variance of $n^{-1/2}\sum_{i=1}^n\hat Z_i(\theta_1^0)$ as
\begin{equation*}
    \hat{\nu}_{1n}^2(\theta_1^0) = n^{-1}\sum_{i=1}^n\hat M_i(\theta_1^0)^2.
\end{equation*}

\section{Variance Component  Analysis Over a Sequence of Responses}\label{sec:global}

In some applications, we are interested in testing whether the variance components are all zero over a set of possibly correlated outcomes. One example of such applications is to test the variance components for the activity distribution based on wearable device data where we are interested in testing the variance component at each of the quantiles $t$ of the activity distribution.  Extending model \eqref{eq:yi}, we assume the following outcome model  at level $t$, 
\begin{equation}\label{eq:yit}
    y_i(t)=X_i\beta^*(t)+r_i(t), \quad i=1,\cdots,n,
\end{equation}
where $ r_i(t)\in R^{n_i}$ is a zero-mean random variable with variance $H_i(\theta^*(t))$. 
We assume that $H_i(\theta^*(t))$ has the same linear structure for each $t$, 
\[
H_i(\theta^*(t)) = \sum_{q=1}^d\theta^*_q(t)\Phi_{iq}, \quad \theta^*(t)=\{\theta^*_1(t),\cdots,\theta^*_d(t)\}^T= \{\theta^*_1(t),\theta^{*T}_{(1)}(t)\}^T.
\]

We are interested in testing the null $H_0: \theta^*_1(t)\equiv \theta_1^0, ~t\in[t_1,t_2]$, where $[t_1,t_2]$ is a pre-defined interval. We propose the following maximally selected empirical likelihood ratio statistic (g\ELR),
\begin{equation}\label{max.test}
\Gamma = \sup_{t\in[t_1,t_2]}\hat{c}_n(\theta_1^0,t)\left\{-2\log \ELR(\theta_1^0,t)\right\},
\end{equation}
where $\hat{c}_n(\theta_1^0,t)\left\{-2\log \ELR(\theta_1^0,t)\right\}$ is the \ELR statistic for the outcome at $t$. 
It can be shown that 
$
\Gamma = \sup_{t\in[t_1,t_2]}S(t)+o_p(1),
$
with 
\begin{equation*}
    S(t)=\begin{cases}
    \frac{\{n^{-1/2}\sum_{i=1}^{n} \hat Z_i(\theta_1^0,t)\}^2}{\hat{\nu}_{1n}^2(\theta_1^0,t)}, & \text{ if }\theta_1^0>0, \\
        \frac{\{n^{-1/2}\sum_{i=1}^{n} \hat Z_i(\theta_1^0,t)\}^2}{\hat{\nu}_{1n}^2(\theta_1^0,t)}
        I\{\sum_{i=1}^n \hat Z_i(\theta_1^0,t)\geq 0\}, & \text{ if }\theta_1^0=0, 
    \end{cases}
\end{equation*}
where 
\begin{align*}
\hat Z_i(\theta_1^0,t) &= \tr\Bigg\{\Phi_{i1}\Bigg(\hat R_i(t)-\Phi_{i1}\theta_1^0-\sum_{q=2}^d\hat\theta_q(t)\Phi_{iq}\Bigg)\Bigg\},\\
\hat{\nu}_{1n}^2(\theta_1^0,t) &= n^{-1}\alpha^{-2}\sum_{i=1}^n\Bigg\langle \hat R_i(t)-H_i((\theta_1^0,\hat{\theta}_{(1)}(t)^T)^T),
\Phi_{i1}-\sum_{q=1}^{d-1}F_q\Phi_{iq+1}\Bigg\rangle^2.
\end{align*}

Assessment of the statistical significance of the statistic  $\Gamma$ defined in \eqref{max.test} is challenging because of the dependence of $\hat Z_i(\theta_1^0,t)$.  We propose a simple way of evaluating its significance by perturbing the  \textsc{el}  statistic. Specifically, 
we apply \eqref{eq:zdm} to rewrite $\sum_{i=1}^n\hat Z_i(\theta_1^0,t)$ as $\sum_{i=1}^n\hat M_i(\theta_1^0,t)$, where $\hat M_i(\theta_1^0,t) = \alpha^{-1}\left\langle\Phi_{i1}-\sum_{q=1}^{d-1}F_q\Phi_{iq+1},\hat R_i(t)-H_i((\theta_1^0,\hat\theta_{(1)}^T(t))^T)\right\rangle$.
We can generate the null distribution of  $\Gamma$ by   perturbing  the test statistic $\Gamma^{(g)}$. Specifically, for each $g~(g=1,\cdots,G)$,  we generate  $\xi_i^{(g)}$  from  i.i.d.\ standard normal distribution and define 
\[
S^{(g)}(t)=
\begin{cases}
    
    \frac{\{n^{-1/2}\sum_{i=1}^{n}\hat M_i(\theta_1^0,t)\xi_i^{(g)}\}^2}{\hat\nu_{1n}^2(\theta_1^0,t)}, & \text{ if }\theta_1^0>0, \\
    \frac{\{n^{-1/2}\sum_{i=1}^{n}\hat M_i(\theta_1^0,t)\xi_i^{(g)}\}^2}{\hat\nu_{1n}^2(\theta_1^0,t)}I\{\sum_{i=1}^n\hat M_i(\theta_1^0,t)\xi_i^{(g)}\geq 0\}, & \text{ if }\theta_1^0=0.
\end{cases}
\]
Define the corresponding perturbed test statistic as $\Gamma^{(g)}=\sup_{t\in[t_1,t_2]}S^{(g)}(t)$.
The following Proportion \ref{prop:2} shows that the perturbed test statistics have the same distribution as the original test statistic under the null. 
Therefore, the $p$-value of $\Gamma$ can be approximated by $\sum_{g=1}^GI(\Gamma^{(g)}>\Gamma)/G$.

\begin{proposition}\label{prop:2}
    $\hat M_i(\theta_1^0,t)$ satisfies
    \begin{enumerate}
        \item[(i)] $E\{n^{-1/2}\sum_{i=1}^n\hat M_i(\theta_1^0,t)\xi_i^{(g)}\}-E\{n^{-1/2}\sum_{i=1}^n\hat Z_i(\theta_1^0,t)\}=o(1)$;
        \item[(ii)] $\var\{n^{-1/2}\sum_{i=1}^n\hat M_i(\theta_1^0,t)\xi_i^{(g)}\}-\var\{n^{-1/2}\sum_{i=1}^n\hat Z_i(\theta_1^0,t)\}=o(1)$;
        \item[(iii)] $\cov\{n^{-1/2}\sum_{i=1}^n\hat M_i(\theta_1^0,s)\xi_i^{(g)},n^{-1/2}\sum_{j=1}^n\hat M_j(\theta_1^0,t)\xi_j^{(g)}\}-\cov\{n^{-1/2}\sum_{i=1}^n\hat Z_i(\theta_1^0,s),$\\
         $n^{-1/2}\sum_{j=1}^n\hat Z_j(\theta_1^0,t)\}=o(1)$.
    \end{enumerate}
\end{proposition}

\section{Simulation studies}\label{sec:simu}
\subsection{Data generation}
We examine the performance of the proposed empirical likelihood ratio tests for variance components and compare the results with  the standard likelihood ratio (\textsc{lr}) test assuming  Gaussian random effects and Gaussian errors. 
To mimic the twin design in the heritability analysis of wearable device data that we  analyze next,  we simulate data on a monozygotic or dizygotic  twin pair. For the $i$th twin, let $n_i=n_{i1}+n_{i2}$, where $n_{i1}$ and $n_{i2}$ are the numbers of repeated measures for the twin. In wearable device data, $y_i(t)$ represents the $t$th quantile of the activity distributions over $n_i$ days. 
The data are generated from the model:
\begin{equation}\label{eq:herit}
    y_i(t) = X_i\beta(t) + T_ia_i(t)+\tau_i(t),\quad i=1,\cdots,n, \quad t=s_1,s_2,\cdots, s_m,
\end{equation}
where $T_i=\text{blkdiag}\{\bone_{n_{i1}},\bone_{n_{i2}}\}$, $a_i(t)$ is a random intercept, and $\tau_i(t)$ denotes zero-mean noise with variance $\sigma_M^2(t)\bI_{n_i}$. Here, $a_i(t)$ is assumed as the sum of additive genetic effect $g_i(t)$, common environment $c_i(t)$, and unique subject-specific environment $e_i(t)$, i.e.,
\[
a_i(t) = g_i(t) + c_i(t) + e_i(t),
\]
where $g_i(t),c_i(t),e_i(t)$ are independent zero-mean random variables with variance-covariance  $\sigma_A^2(t)K_i$, $\sigma_C^2(t)\Lambda_i$, and  $\sigma_E^2(t)\bI_2$, respectively. The variance components $\sigma_A^2(t),\sigma_C^2(t)$, and $\sigma_E^2(t)$ represent the additive genetic variance, common environmental variance, and unique environmental variance, respectively. 
For the $i$th twin, $K_i$ is a genetic similarity matrix with $$K_i=\begin{pmatrix} 1 & 1 \\ 1 & 1 \end{pmatrix} \mbox{ for monozygotic  twin}, \mbox{    } K_i=\begin{pmatrix} 1 & 0.5 \\ 0.5 & 1 \end{pmatrix} \mbox{ for dizygotic  twin}, 
$$ and 
 $\Lambda_i$ quantifies shared environment between the twin pair with $$\Lambda_i=\begin{pmatrix} 1 & 1 \\ 1 & 1 \end{pmatrix}.$$ Under this model, we have 
\begin{align*}
H_i(\theta^*(t)) &= \sigma_A^2(t)T_iK_iT_i^T+\sigma_C^2(t)T_i\Lambda_iT_i^T+\sigma_E^2(t)T_iT_i^T+\sigma_M^2(t)\bI_{n_i},\\
~\theta^*(t)&=(\sigma_A^2(t),\sigma_C^2(t),\sigma_E^2(t),\sigma_M^2(t))^T.
\end{align*}

We sample $n_{ik}~(i=1,\cdots,n; k=1,2)$ from $\{3,4,5,6,7\}$ with equal probability $1/5$. 
We set $n=100$, among which there are 50 monozygotic twin families and 50 dizygotic twin families, and $t=0.01,0.03,0.05,\cdots,0.99$.
Denote by $x_{ij}^T$ the $j$th row of $X_i$.
Let $x_{ij}=(1,x_{ij1},x_{ij2})^T$ with $x_{ij1}\sim N(0,1)$ and $x_{ij2}\sim N(2,1)$. Moreover, we set 
$\beta(t)=(\beta_1(t),\beta_2(t),\beta_3(t))^T$ where $\beta_1(t)$ and $\beta_3(t)$ are the quantile functions of $N(1,6)$ and $N(1,9)$, respectively, and $\beta_2(t)=0$. Let
\[
\sigma_A^2(t) = \sum_{l=1}^{N_a}\lambda_l^a(t)(\psi_l^a(t))^2,\quad \sigma_C^2(t)=0,\quad\sigma_E^2(t) = \sum_{l=1}^{N_e}\lambda_l^e(t)(\psi_l^e(t))^2, 
\]
where $N_a=N_e=2$, $(\lambda_1^a(t),\lambda_2^a(t))=C_a(t)(0.5,1)$, $(\lambda_1^e(t),\lambda_2^e(t))=C_e(t)(0.6,0.9)$, $\psi_1^a(t)=\psi_2^e(t)=\sqrt{2}\sin(2\pi t)$, 
$\psi_2^a(t)=\psi_1^e(t)=\sqrt{2}\cos(2\pi t)$. 

\subsection{\ELR~test for single variance component}
For each $t\in\{0.01,0.03,0.05,\cdots,0.99\}$, we test the null hypothesis $H_0:\sigma_A^2(t)=0$.
We consider the following two types of distributions for $g_i(t),c_i(t),e_i(t)$ and $\tau_i(t)$:
\begin{enumerate}
    \item [(i)] multivariate normal distribution, $g_i(t)\stackrel{iid}{\sim}\cN(\bzero,\sigma_A^2(t)K_i)$, $c_i(t)\stackrel{iid}{\sim}\cN(\bzero,\sigma_C^2(t)\Lambda_i)$, $e_i(t)\stackrel{iid}{\sim}\cN(\bzero,\sigma_E^2(t)\bI_i)$, $\tau_{i}(t)\stackrel{iid}{\sim} \cN(\bzero,0.3\bI_{n_i})$;
    \item [(ii)] multivariate $t$ distribution, $g_i(t)\stackrel{iid}{\sim}t_3(\bzero,\sigma_A^2(t)K_i/3)$, $c_i(t)\stackrel{iid}{\sim}t_3(\bzero,\sigma_C^2(t)\Lambda_i/3)$, $e_i(t)\stackrel{iid}{\sim}$\\$t_3(\bzero,\sigma_E^2(t)\bI_i/3)$, $\tau_{i}(t)\stackrel{iid}{\sim} t_3(\bzero,0.1\bI_{n_i})$.
\end{enumerate}

\begin{figure}[!h]
    \centering
    \begin{tabular}{c}
        \includegraphics[width=1.0\textwidth]{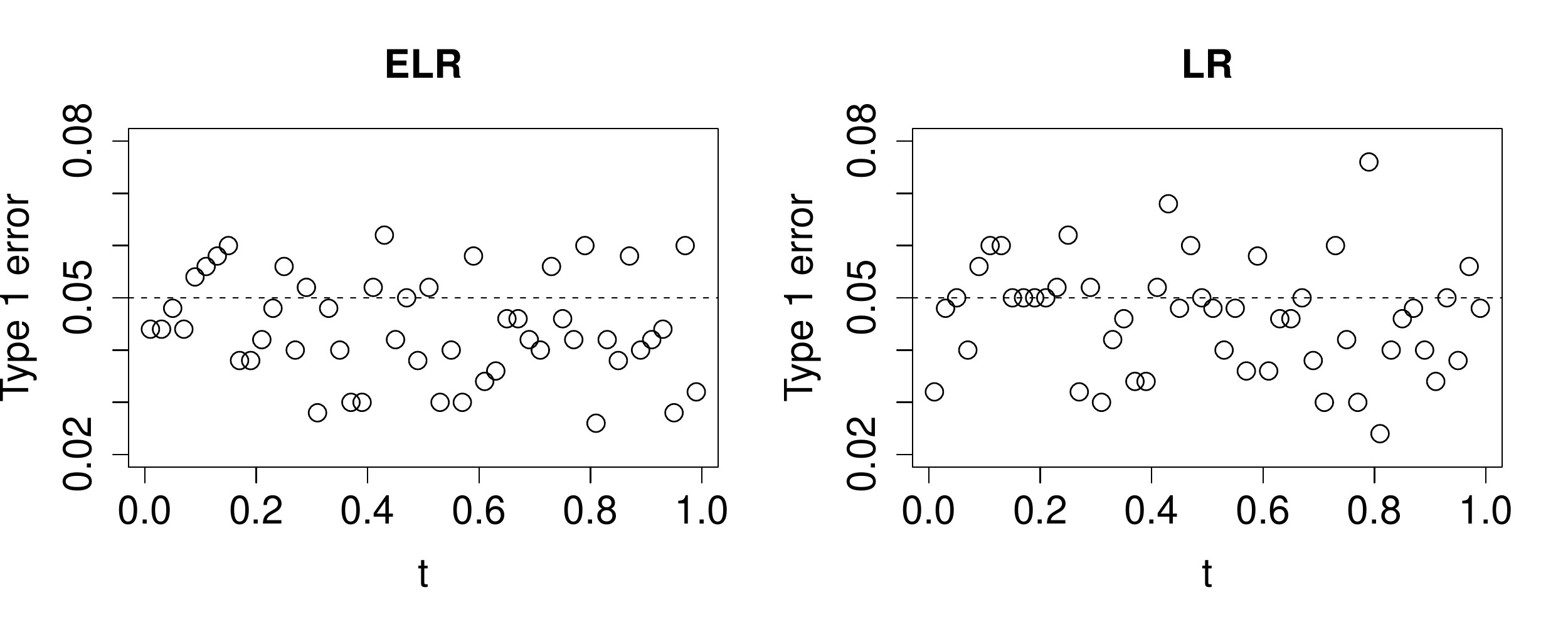}\\
        (a) normal random effects and errors. 
        \\
        \includegraphics[width=1.0\textwidth]{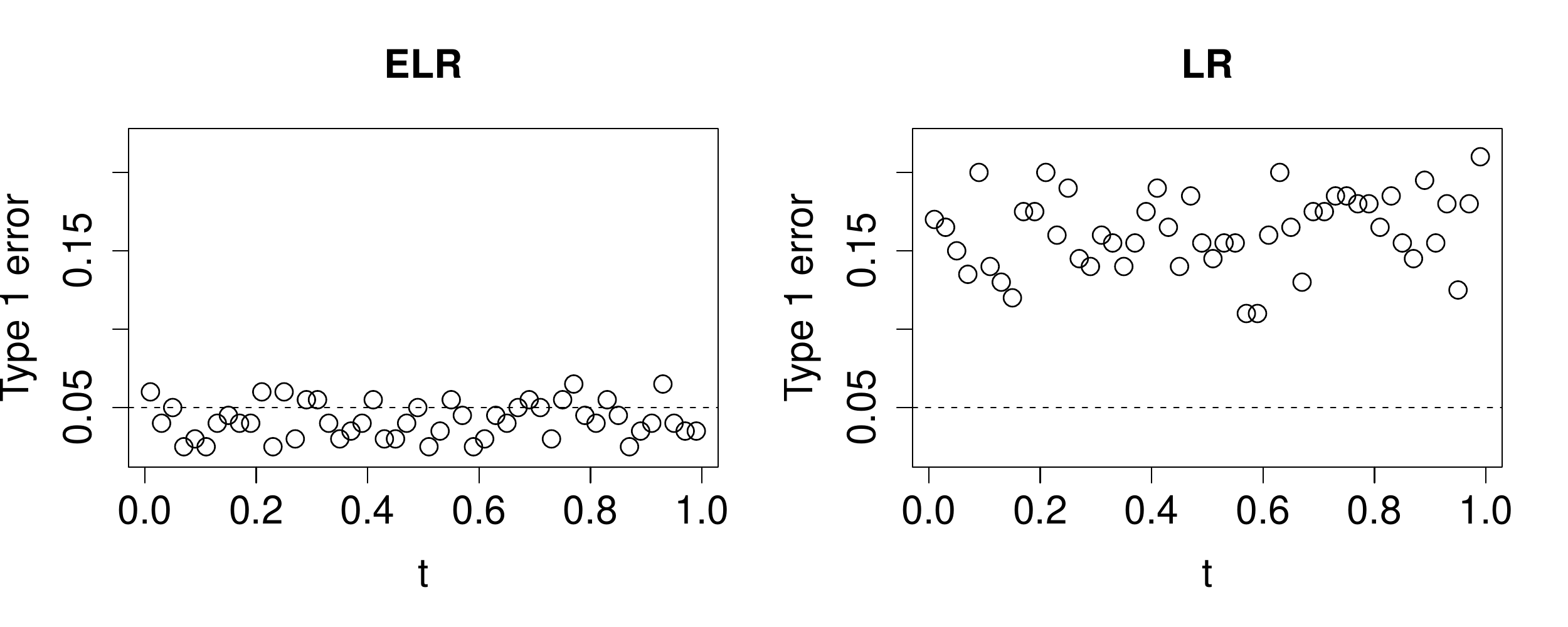}\\
        (b) $t$-distributed random effects and errors.  
    \end{tabular}
    \caption{Type 1 error for each given value of $t$.  Top: normal random effects and errors;  bottom: $t$-distributed random effects and errors. 
        $\ELR$:  \textsc{el}  ratio test with the lest-square estimate of $\beta^*$;  \textsc{lr} : likelihood ratio test under the normal random-effects assumptions. }\label{fig:herit1}
\end{figure}

Denote by \ELR~the proposed empirical likelihood ratio test with unknown $\beta^*(t)$.
We first consider the type 1 error for each given value of $t$. 
We consider the setting $C_a(t)=0$ and $C_e(t) = 0.1$  with random errors generated both from a normal and a $t$ distribution. We  repeat the simulations 500 times.  Figure \ref{fig:herit1} gives the results of type 1 errors for different values of $t$.  We see all the methods perform well under the normal errors. However,  \textsc{lr}  shows inflated type 1 errors when the error follows a long-tailed $t$ distribution.

To evaluate the power of the proposed tests, we consider the model with  $C_a(t)=0.1$ and $C_e(t) = 0.1$.
We calculate the empirical power of the proposed test at 0.05 level for different values of $t$ and present the results in Figure \ref{fig:herit2}. The proposed method exhibits similar power as the  \textsc{lr}  test under the normal error. 
When the error follows a $t$ distribution, we do not report  the result of the  \textsc{lr}  test because of its inflated type 1 error as shown in Figure \ref{fig:herit1} (b), and the \ELR~test does not lose much power. 

\begin{figure}[!h]
    \begin{center}
\includegraphics[width=1\textwidth]{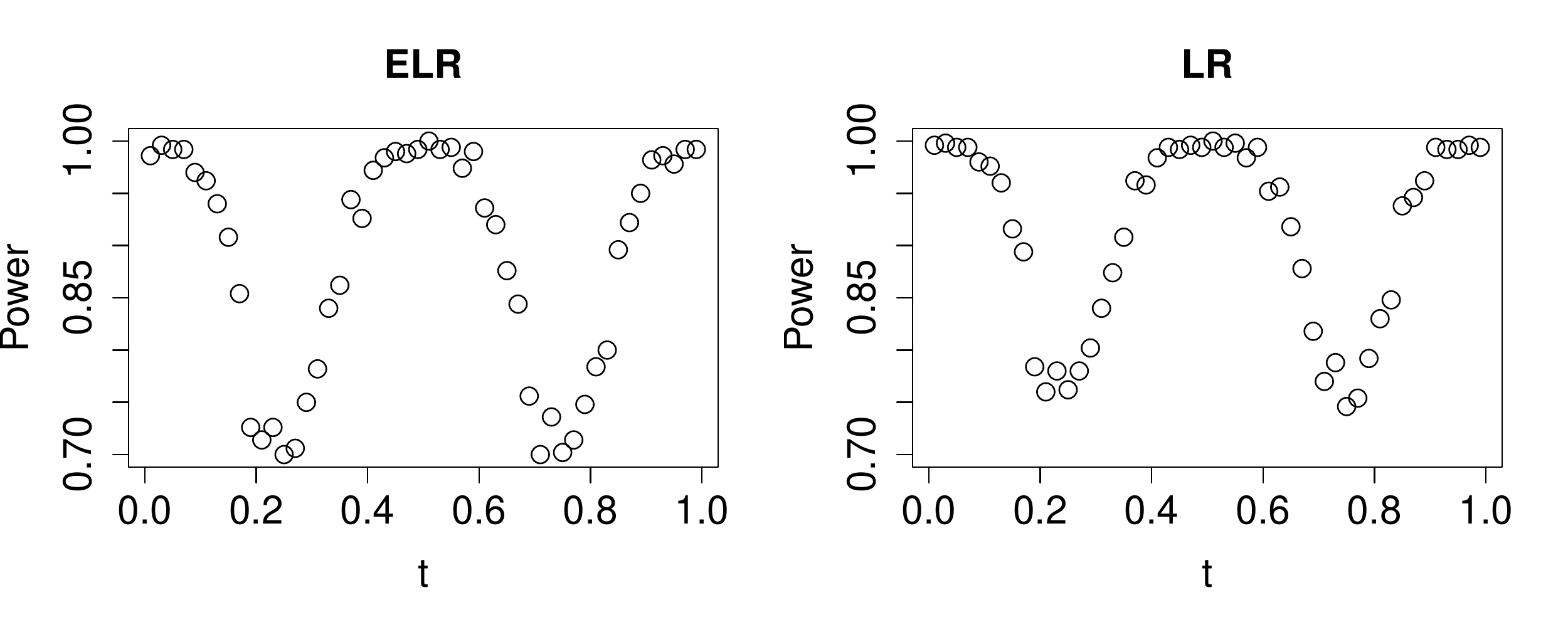}\\
        (a) normal random effects and errors. 
        \\
        \includegraphics[width=0.5\textwidth]{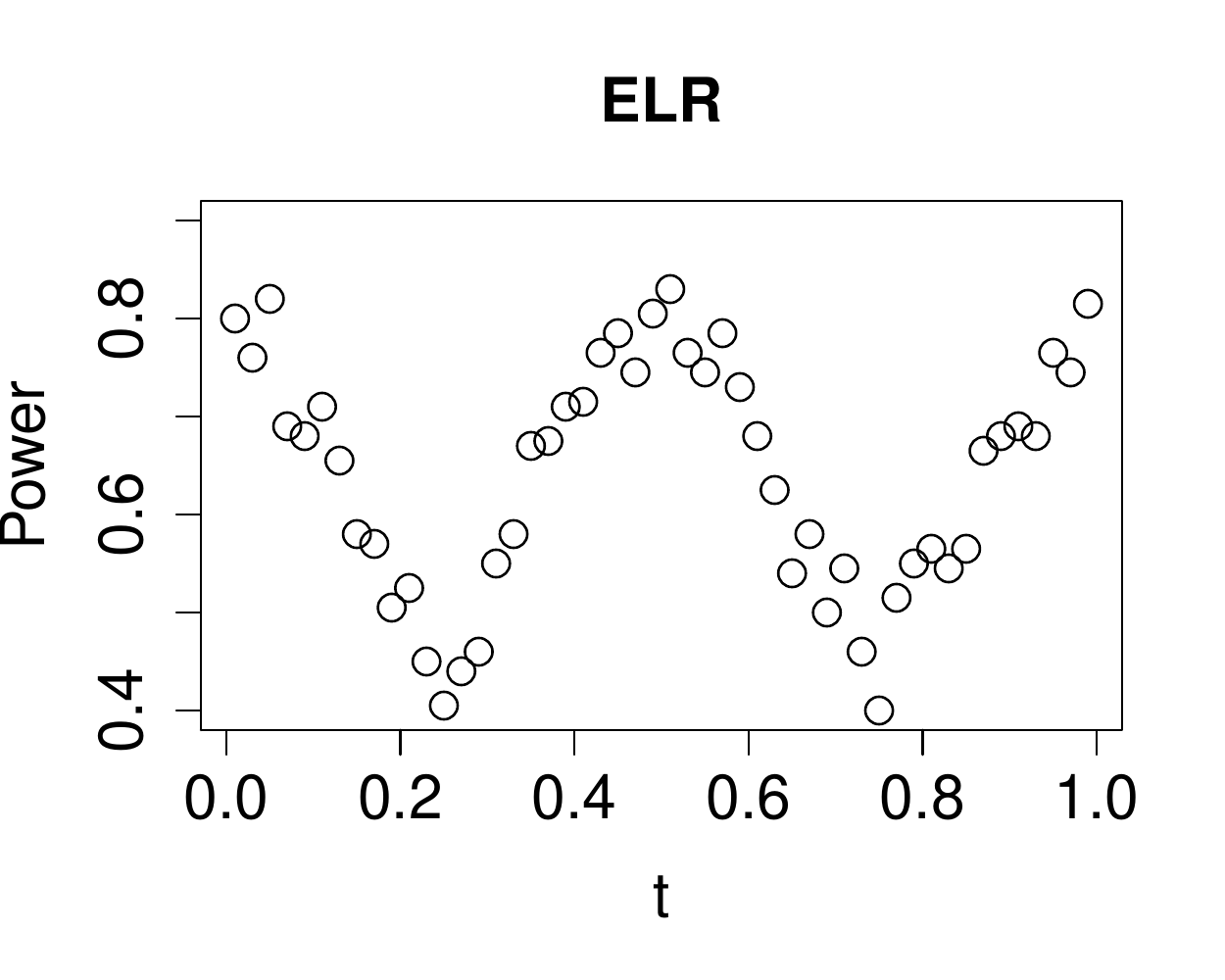}\\
        (b) $t$-distributed random effects and errors. 
\end{center}
    \caption{Empirical power for each given value of $t$. Top: normal random effects and errors;  bottom: $t$-distributed random effects and errors.    
    $\ELR$:  \textsc{el}  ratio test with the least-square estimate of  $\beta^*$;  \textsc{lr} : likelihood ratio test under the normal random-effects assumptions.}\label{fig:herit2}
\end{figure}

\subsection{\ELR~test for variance components over an interval}
To evaluate the proposed tests for variance components in the case of correlated outcomes over an interval of $t$, we generate data as follows.
Let $g_i(t)=\sigma_A(t)\zeta_{ai},c_i(t)=\sigma_C(t)\zeta_{ci}$, and $e_i(t)=\sigma_E(t)\zeta_{ei}$.
Let $\tau_i(t)=\sigma_M(t)\zeta_{\tau i}$, where
$
\sigma_M^2(t) = \sum_{l=1}^{N_m}\lambda_l^m(t)(\psi_l^m(t))^2
$
with $N_m=2$, $(\lambda_1^m(t),\lambda_2^m(t))=C_m(t)(0.5,1)$, $\psi_1^m(t)=\sqrt{2}\cos(2\pi t)$, and $\psi_2^m(t)=\sqrt{2}\sin(2\pi t)$. We consider two types of distributions for $\zeta_{ai},\zeta_{ci},\zeta_{ei}, \zeta_{\tau i}$: 
\begin{enumerate}
    \item [(i)] multivariate normal distribution, $\zeta_{ai}\stackrel{iid}{\sim}\cN(\bzero,K_i)$, $\zeta_{ci}\stackrel{iid}{\sim}\cN(\bzero,\Lambda_i)$, $\zeta_{ei}\stackrel{iid}{\sim}\cN(\bzero,\bI_2)$, $\zeta_{\tau i}\stackrel{iid}{\sim} \cN(\bzero,\bI_2)$;
    \item [(ii)] multivariate $t$ distribution, $\zeta_{ai}\stackrel{iid}{\sim}t_3(\bzero,K_i/3)$, $\zeta_{ci}\stackrel{iid}{\sim}t_3(\bzero,\Lambda_i/3)$, $\zeta_{ei}\stackrel{iid}{\sim}t_3(\bzero,\bI_2/3)$, $\zeta_{\tau i}\stackrel{iid}{\sim} t_3(\bzero,\bI_2/3)$.
\end{enumerate}

\begin{figure}[!h]
    \centering
    \includegraphics[width=0.85\textwidth]{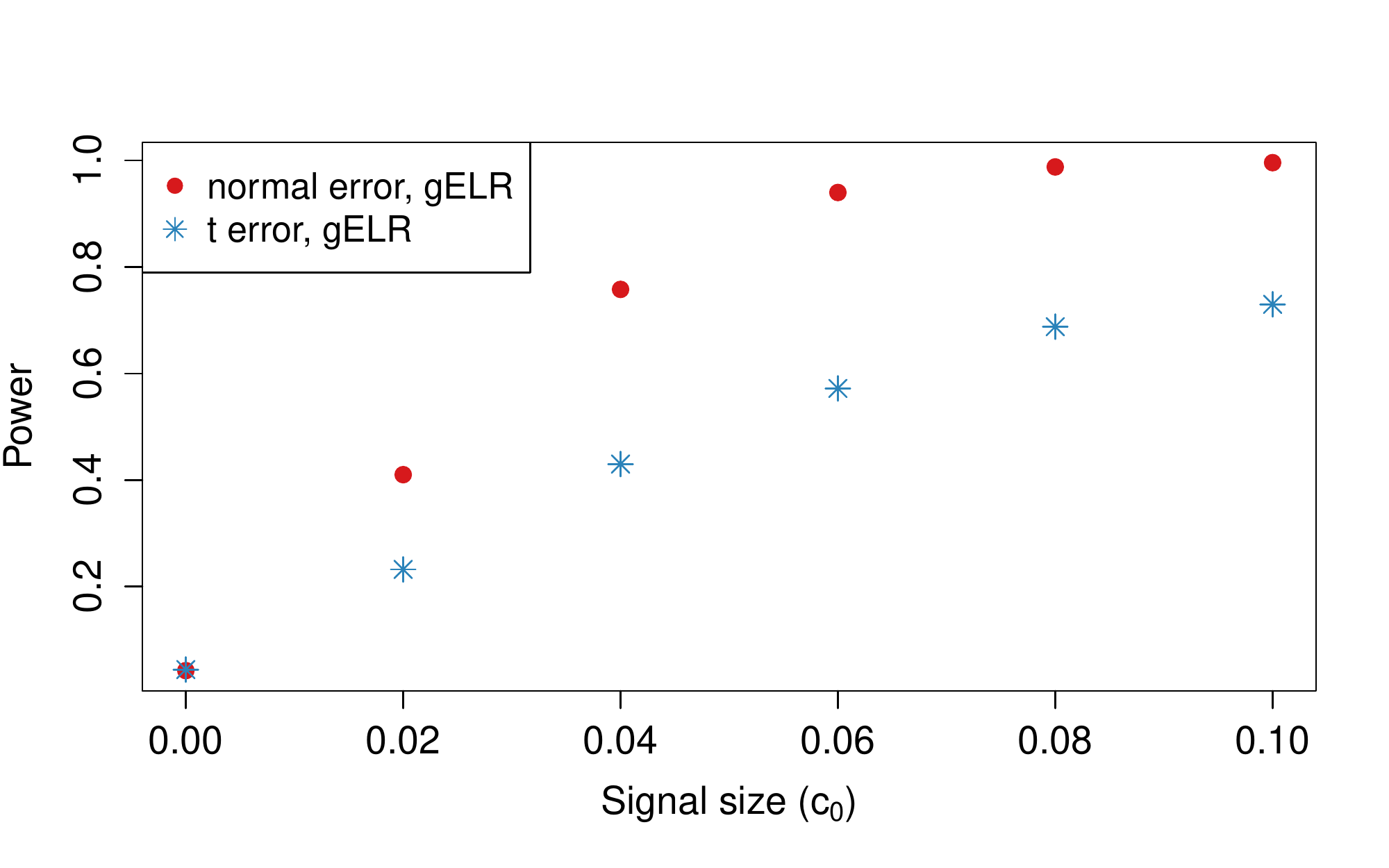}
    \caption{Empirical power for the global test $H_0:\sigma_A^2(t)\equiv 0,~t\in[0,1]$ at different values of $c_0$.}\label{fig:simuglobal}
\end{figure}

Denote by gELR the proposed global empirical likelihood ratio test with unknown $\beta^*(t)$.
We first consider the global test $H_0:\sigma_A^2(t)\equiv 0,~t\in[0,1]$. Let $C_a(t)=c_0I(t=0.49)$, $C_e(t)=0.1$, and $C_m(t)=0.08$, where $I(\cdot)$ is an indicator function. We consider different choices of the signal size $c_0$ by setting $c_0=0, 0.02, 0.04, 0.06, 0.08, 0.1$, and generate 500 datasets for each setting. 
Figure \ref{fig:simuglobal} presents the empirical power of gELR at 0.05 significance level under different distributions of random errors and different $c_0$. As expected, the empirical power of rejecting the null hypothesis increases with the signal size. 
Compared to the results under the multivariate $t$ distribution, the proposed test gELR has higher power when data are normally distributed.

To further evaluate the type 1 error and the power, we consider models with 
 $C_a(t)=0.08I(t\in\{0.47,0.49,0.51,0.53\})$, $C_e(t)=0.1$, and $C_m(t)=0.08$. We consider to test each of the 
candidate intervals of lengths  $\{3,4,5,6\}$  and denote them by scan3, scan4, scan5, and scan6, respectively. 
Let $\cJ_k$ be the set of candidate intervals under the scanning length $k$ ($k=3,4,5,6$) and let $\cJ=\cup_{k=3}^6\cJ_k$ be the set of all candidate intervals.
For each candidate interval $L\in\cJ$, we test the null hypothesis $H_0:\sigma_A^2(t)\equiv 0,~t\in L$. 
The signal in the interval $L$ is significant if 
\begin{equation*}\label{eq:h}
    h(\Gamma_L) = \frac{\Gamma_L-\bar\Gamma_L}{\sqrt{\sum_{g=1}^G(\Gamma_L^{(g)}-\bar\Gamma_L)^2/(G-1)}}> \sqrt{2\log|\cJ|},
\end{equation*}
where $\bar\Gamma_L=(\sum_{g=1}^G\Gamma_L^{(g)})/G$ with $G=1000$. The threshold $\sqrt{2\log|\cJ|}$ is selected based on the extreme value distribution of $|\cJ|$ normal random variables.

Under each type of error distributions, 500 datasets are generated.
For the global test under the candidate interval $L=\{t_1,t_1+0.02,\cdots,t_2\}$, we mark its empirical power at $(t_1+t_2)/2$.
The results are shown in Figure \ref{fig:herit3}.
The proposed global test gELR exhibits high power if the interval involves at least one nonzero time points and show almost no power otherwise.

\begin{figure}[!h]
    \begin{center}
    \begin{tabular}{c}
        \includegraphics[width=.5\textwidth]{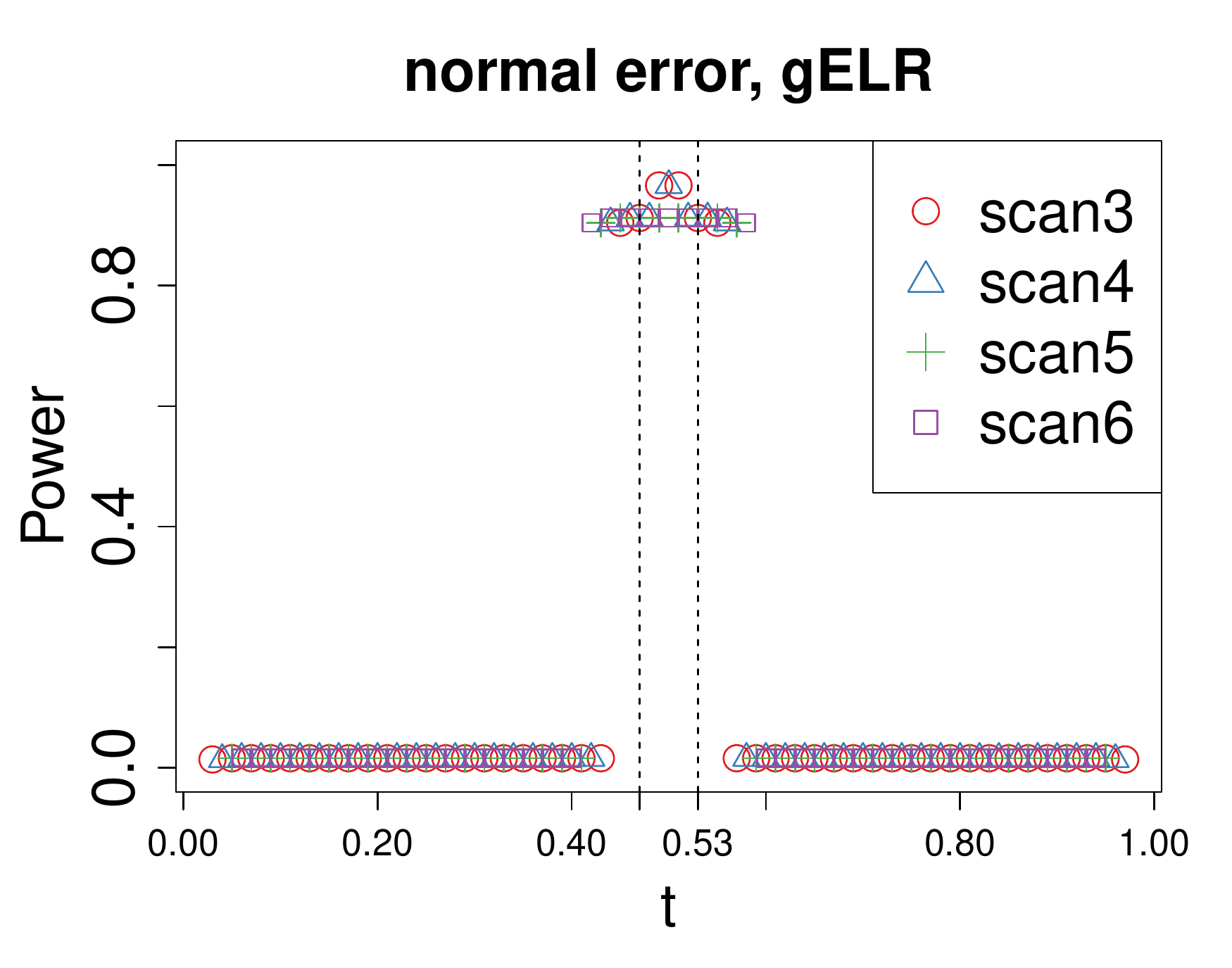}
        \includegraphics[width=.5\textwidth]{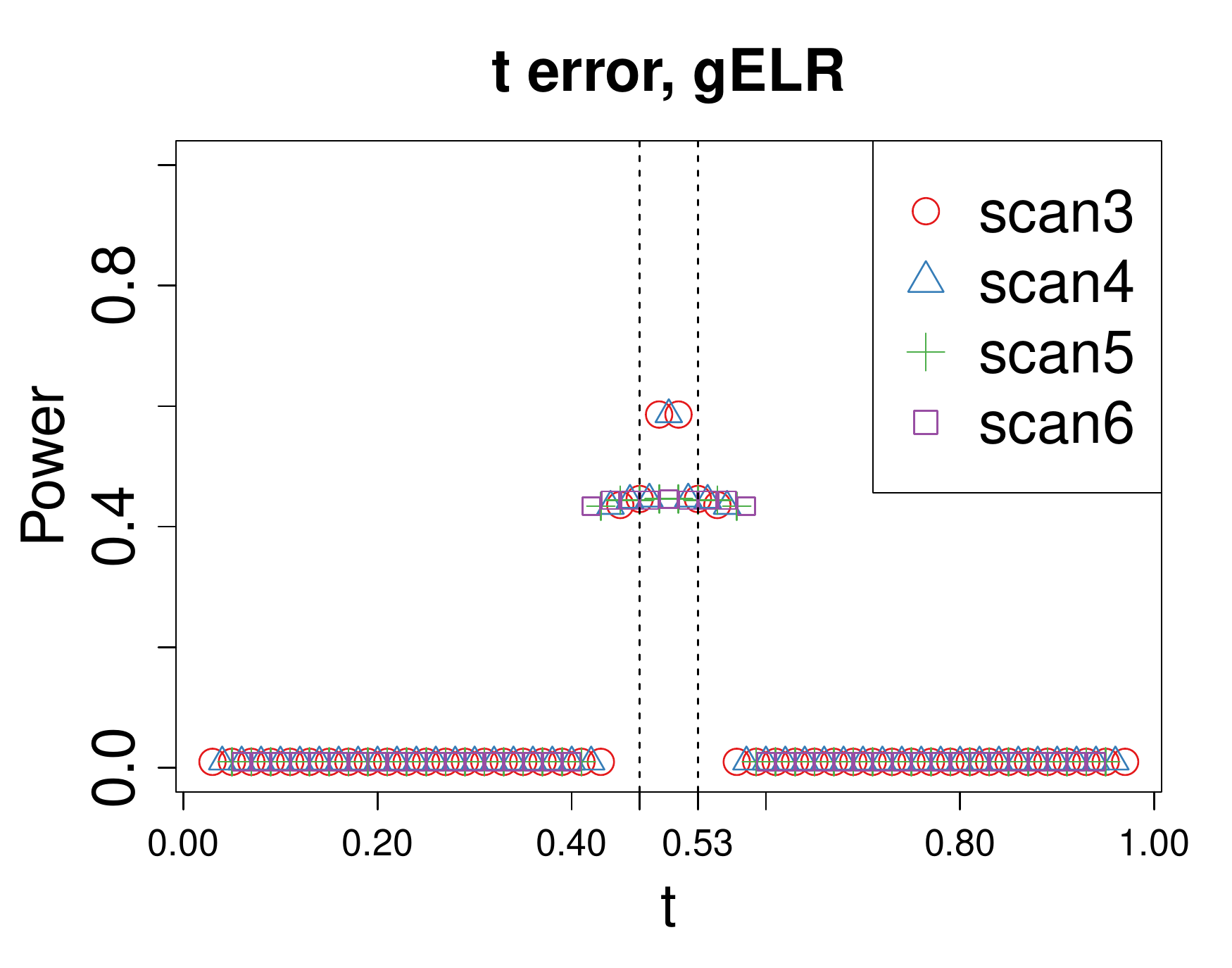}
    \end{tabular}
\end{center}
    \caption{Empirical  power  of testing zero variance component in a given interval of length of 3,4,5 and 6, where the true non-zero variance compnents are at \{0.47, 0.49, 0.51, 0.53\}. Left: normal random effect and error; right: $t$-radnom effect and error.}\label{fig:herit3}
\end{figure}

\section{Application to genetic heritability analysis of physical activity distribution}\label{sec:realapp}

\subsection{Description of the data}
We apply the methods to data set of the  Australian Twin study, which includes 366 healthy twins, 151 of them are monozygotic  twins, and 215 are dizygotic  twins. The participants wore actigraphy to track their physical activities for no more than 14 days. The minute-to-minute activity counts derived from actigraphy were collected in a 1440-dimensional vector per day. Since we are interested in inference of the heritability of the activity distributions, we obtain the empirical quantiles of activity counts at different quantiles, $t=1/144,~2/144,\cdots,144/144$. Specifically,  for the $j$th measurement (day) from the $k$th person in the $i$th twin family, the raw data $\xi_{ikj}=(\xi_{ikj1},\cdots,\xi_{ikj1440})^T$ from the wearable device are transformed by using
\[
\tilde\xi_{ikj} = \log(9250\xi_{ikj}+1), \quad i=1,\cdots,n;~ k=1,2;~ j=1,\cdots,n_{ik}.
\]
For the $k$th person in the $i$th twin family, the $j$th repeated measure of $t$-quantile of activity counts is  obtained as
\[
y_{ikj}(t) = \tilde\xi_{ikj}^{[1440\cdot t]},\quad t=1/144,~2/144,\cdots,144/144,
\]
where $\tilde\xi_{ikj}^{[s]}$ denotes the $s$th order statistic of $\tilde\xi_{ikj}$. 

The covariate $x_{ikj}$ includes gender, age, BMI, and indicator of weekend, i.e., $x_{ikj}=(1,\text{Gender},\text{Age},$ $\text{BMI},\text{Weekend})^T$. 
Let $y_i(t)=(y_{i11}(t),\cdots,y_{i1n_{i1}}(t),y_{i21}(t),\cdots,y_{i2n_{i2}}(t))^T$ and $X_i=(x_{i11}^T,\cdots,x_{i1n_{i1}}^T,$ $x_{i21}^T,\cdots,x_{i2n_{i2}}^T)^T$.
Let $y(t)=(y_1^T(t),\cdots,y_n^T(t))^T$. 
For each $t$, we remove the outliers of $y(t)$ defined as values that are more than 1.5 times the inter-quantile range above the upper quantile or below the lower quantile.
After removing all outliers and removing missing data, we have  $n=149$ twin families including 63 monozygotic twin families and 86 dizygotic twin families, and the total number of observations is 3,489.

\subsection{Effects of Gender, Age, BMI, Weekend on Activity Profiles}
We first examine the associations between  the covariates including gender, age, BMI, and weekend vs weekday and the overall activity distribution. 
For each of the four covariates and each of the $t$ values, we obtain the \textsc{el} estimator by solving the estimating equations $\sum_{i=1}^n\phi_i(\beta) = 0$, and apply Theorem \ref{th:beta1} to construct the confidence interval $\{\beta_0:-2\log\ELR_0(\beta_0)\leq \chi_1^2(1-\alpha)\}$, where $\chi_1^2(1-\alpha)$ is the $(1-\alpha)$ quantile of the $\chi_1^2$ distribution.
The first column of Figure \ref{fig:realbeta} shows the estimated regression coefficient for each of the $t$ values and its point-wise 95\% confidence intervals using the \textsc{el} method.  

\begin{figure}[!ht]
    \centering
        \includegraphics[width=1.0\textwidth]{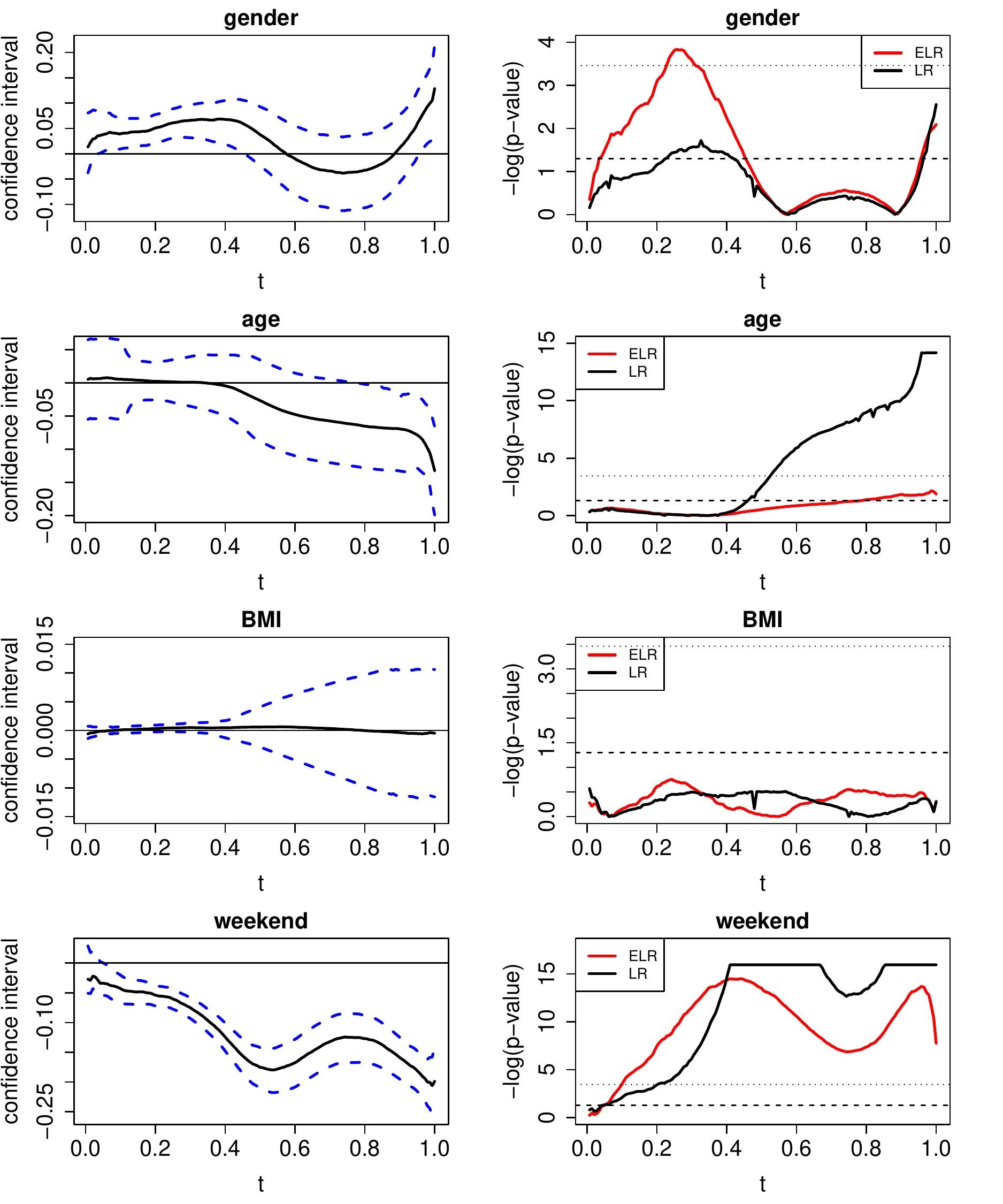}
    \caption{Estimate of $\beta^*(t)$, its 95\% confidence interval (left panel) and the $-\log_{10}(p$-value) (right panel)  for gender, age, BMI, and weekend for each of  quantile $t$ values. 
    The black dashed horizontal line is  the nominal  threshold $-\log_{10}(0.05)$, and the black dotted horizontal line is the Bonferroni corrected threshold $-\log_{10}(0.05/144)$.}\label{fig:realbeta}
\end{figure}

We then test whether there is any difference in activity profiles between individuals of different gender, age, and BMI and whether the activity profiles are different between weekdays and weekends. 
Specifically, we consider testing such differences at each of the quantile $t$.  To test $H_0:\beta_l(t)=0,~l\in\{\text{Gender, Age, BMI, Weekend}\}$, we apply the empirical likelihood  ratio test  in Section \ref{sec:coef} and the standard  likelihood ratio (\textsc{lr}) test assuming normal random effects, and we obtain the $p$-value for each of the $t$ values. 
The second column of Figure \ref{fig:realbeta} shows the $p$-values for each $t$ and for each of these four covariates.  
At the nominal $p$-value of 0.05, the \ELR~test shows that there is an gender effects when the activity counts are small (i.e., small $t$). In contrast, the standard  \textsc{lr}  test only shows such significance in a smaller interval from 0.23 to 0.42. For age, the \ELR~shows a significant effect for the large activity counts region (i.e., large $t$). 
Both the \ELR~and  \textsc{lr}  tests do not reject the null hypothesis that there is no effect of BMI, while the effects of weekend are statistically significant under almost the whole region of $t$. 


\subsection{Analysis of heritability of the activity distribution}
We then address the question whether the activity distribution is heritable, where the distribution  is summarized as the quantiles. This is equivalent to test the null hypothesis $H_0:\sigma_A^2(t)= 0,~t\in[0,1]$. 
For each quantile $t$, we first estimate the fixed effects using the least-square estimate  and then apply the proposed \ELR~to test the null hypothesis $H_0:\sigma_A^2(t)= 0$ and to compare the results  with the  \textsc{lr}  method. Figure \ref{fig:realherit} (a) gives the $p$-values at different  quantiles $t$.  It shows that the  test $H_0:\sigma_A^2(t)=0$ is rejected for  $t\in [0.375,0.958]$ based on the \ELR~test and $t\in[0.472,0.931]$ using  \textsc{lr}  at the nominal  0.05 significance level.  However, if we use the Bonferroni correction for multiple testing, only the proposed \ELR~test identifies significant heritability for the quantiles between 0.375 and 0.514. 
The $p$-value of global test $H_0:\sigma_A^2(t)\equiv 0,~t\in[0,1]$ is 0 when applying the proposed \ELR~with 1000 permutations.   Overall, our analysis shows that the activity distribution is heritable, especially in the quantile range from 0.375 to 0.514.

\ignore{
We finally apply the scanning procedure in order to identify the quantile intervals where the activity quantiles are heritable. We set the  the scanning lengths to $8,9,10,11,12$, and apply gELR to the candidate intervals. Since the scanning length may be less than the true length of signal interval, we combine the identified intervals that are very close to each other as recommended in \cite{jeng2010optimal}. This procedure results in identifying the signal interval $t\in[0.354,0.903]$ where the activity quantiles are inheritable, which is very close to the results from \ELR~ for each of the quantile $t$. 
}

\begin{figure}[!h]
    \centering
    \begin{subfigure}{.5\textwidth}
    \centering
    \includegraphics[width=\linewidth]{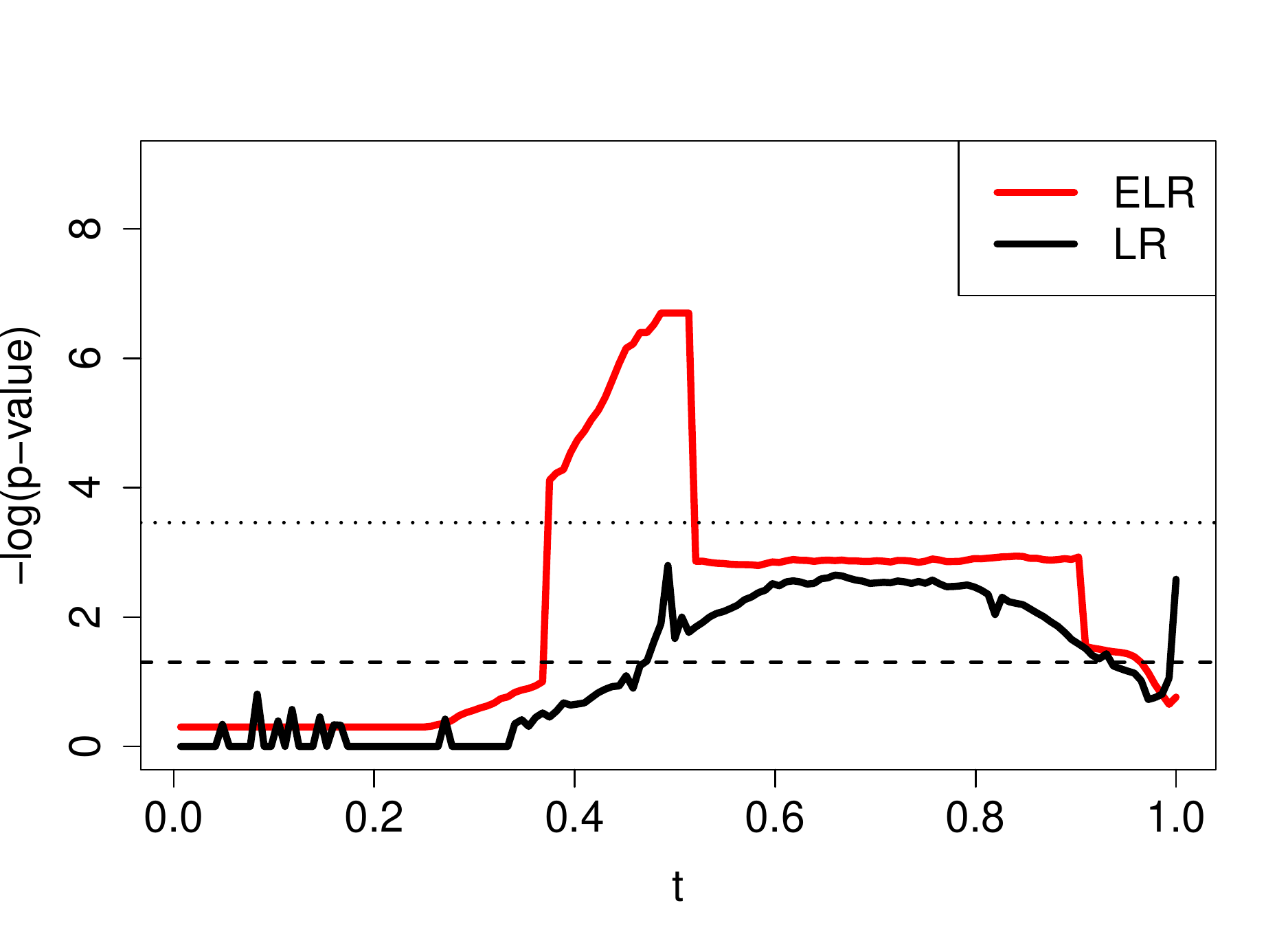}
    \caption{Results under preprocessing 1}
    \end{subfigure}%
    \begin{subfigure}{.5\textwidth}
    \centering
    \includegraphics[width=\linewidth]{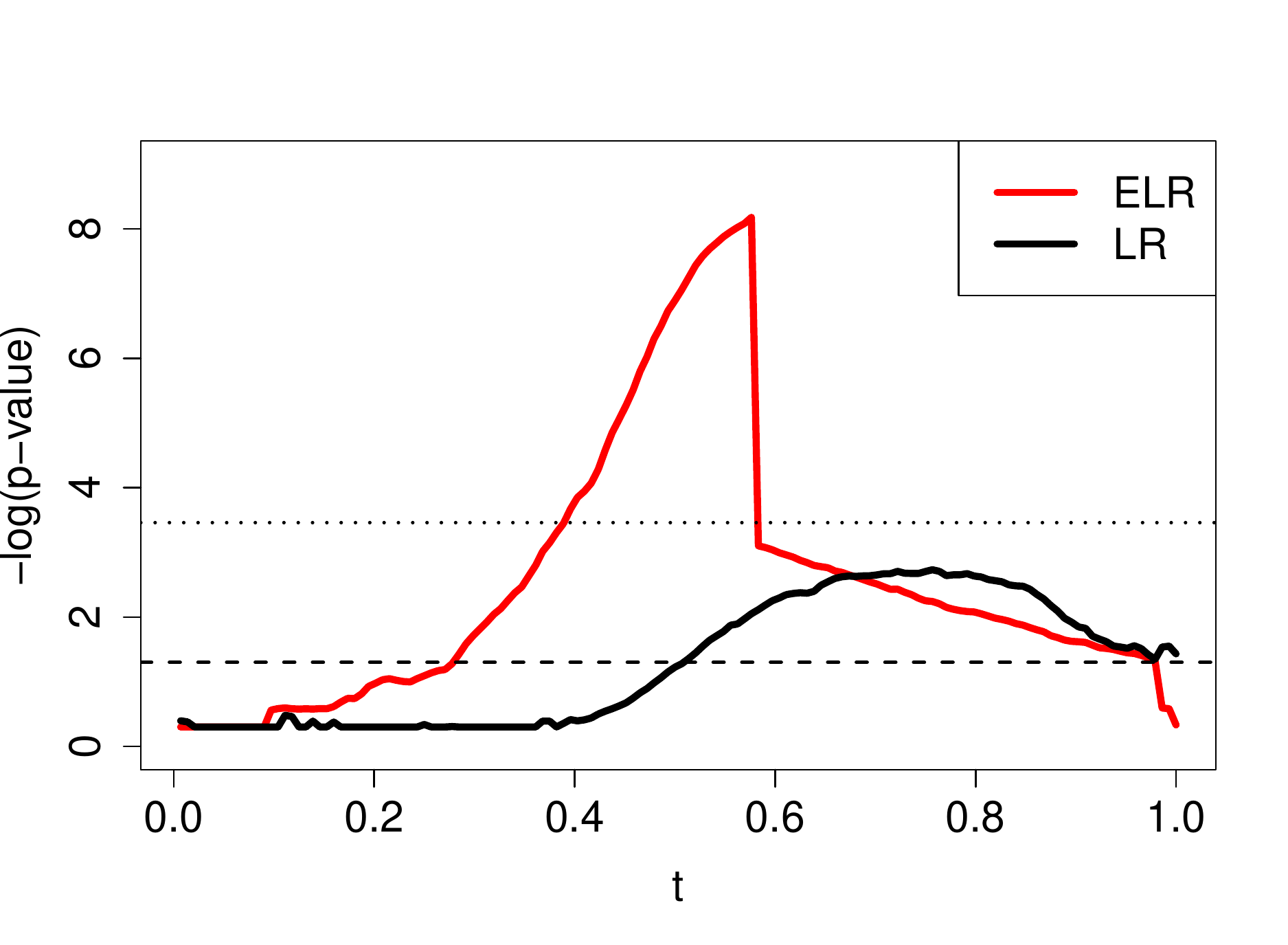}
    \caption{Results under preprocessing 2}
    \end{subfigure}

    \caption{The $-\log_{10}(p\text{-value})$ of testing heritability $H_0:\sigma_A^2(t)=0$ for different quantile values $t$. 
    The black dashed horizontal line is  the nominal  threshold $-\log_{10}(0.05)$, and the black dotted horizontal line is the Bonferroni corrected threshold $-\log_{10}(0.05/144)$.
    Left: results under preprocessing 1; right: results under preprocessing 2.
    }\label{fig:realherit}
\end{figure}

\subsection{Sensitivity analysis of heritability of the activity distribution}
To examine whether our previous preprocessing  steps  affect the analysis of heritability, we consider an alternative approach  to remove  the outliers. 
For each $t$, we remove the outliers of $y(t)$ defined as the  values that are greater than  3 standard deviations from its median.  After removing all the outliers and missing data, we have $n = 152$ twin families including 64 monozygotic twin families and 88 dizygotic twin families, and the total number of observations is 4,190. 

Under this  preprocessing method,   the $-\log_{10}(p$-value) of testing heritability is provided in Figure \ref{fig:realherit} (b), which shows  similar results as Figure \ref{fig:realherit} (a). 
The test $H_0:\sigma_A^2(t)=0$ is rejected for $t\in[0.285,0.979]$ when using the \ELR test and $t\in[0.514,1]$ using the \textsc{LR} test at the nominal 0.05 significance level. 
If we adopt the Bonferroni correction for multiple testing, only the proposed \ELR test  identified significant heritability under the quantiles between 0.396 and 0.576. 
The $p$-value of the  global test $H_0:\sigma_A^2(t)\equiv 0,~t\in[0,1]$ is 0 when applying the proposed \ELR with 1000 permutations.

\section{Discussion}\label{sec:discuss}
In this paper, we have developed an  empirical likelihood method for making inference of the variance components in  general linear mixed-effects models. The proposed empirical likelihood ratio test statistic can be applied to a large set of related outcomes such as different quantiles of the activity distribution when we analyze the wearable device data sets. 
Simulation studies show that the proposed methods control type 1 error much better than the likelihood ratio method when the normality assumptions do not hold. 

To address the unknown nuisance variance components, we assume its true value $\theta^*_{(1)}$ being  positive and thus as $n\rightarrow\infty$,  \eqref{eq:theta2} provides unbiased positive estimates  with probability 1. When applying the proposed methods to the real data, we note that  (\ref{eq:theta2}) may provide  negative estimators at some quantile $t$.  To solve this problem,  we first test whether the nuisance variance component (for example, $\sigma_C^2(t)$ or $\sigma_E^2(t)$) is zero at these quantile  points. If the null hypothesis is not rejected, we omit the nuisance variance components in the model and then apply the proposed \ELR test for the components of interest.

\section*{ACKNOWLEDGMENT}
This research was supported by the Intramural Research Program of the National Institute of Mental Health through grant ZIA MH002954-04 [Motor Activity Research Consortium for Health (mMARCH)]. We thank Dr. Hickie and Dr. Martin for sharing the Australian twin study data as part of the mMARCH network and the Genetic Epidemiology Research Branch at National Institute of Mental Health for processing the accelerometry data.

\section*{APPENDIX}
\appendix

\renewcommand{\thetheorem}{A\arabic{theorem}}
\renewcommand{\thelemma}{A\arabic{lemma}}
\renewcommand{\theequation}{A\arabic{equation}}

\setcounter{theorem}{0}
\section{Proofs and complements}
\subsection{Proof of Theorem \ref{th:3}}
    To prove Theorem \ref{th:3}, we first consider the setting with known $\beta^*$.
We define $Z_i(\theta_1)$, $L_1(\theta_1)$, $\ELR_1(\theta_1^0)$, and $\tilde \nu_{1n}^2(\theta_1^0)$ in the same way as $\hat Z_i(\theta_1)$, $L(\theta_1)$, $\ELR(\theta_1^0)$, and $\hat \nu_{1n}^2(\theta_1^0)$, respectively, with $\hat R_i$ replaced by $R_i$.
Under Conditions  \ref{cond:1-2} and \ref{cond:1-1}, we derive the asymptotic distribution of $\ELR_1(\theta_1^0)$ in the following theorem.
\begin{theorem}\label{th:1}
    Let $\tilde c_n(\theta_1^0)=\tilde\nu_{2n}^2(\theta_1^0)/\tilde\nu_{1n}^2(\theta_1^0)$, where $\tilde\nu_{1n}^2(\theta_1^0)$ is a consistent estimator of the asymptotic variance of $n^{-1/2}\sum_{i=1}^nZ_i(\theta_1^0)$ and $\tilde\nu_{2n}^2(\theta_1^0)=n^{-1}\sum_{i=1}^nZ_i^2(\theta_1^0)$. 
    If $\theta^*_{(1)}\in R_+^{d-1}$, then under Conditions \ref{cond:1-1} and \ref{cond:1-2}, as $n\rightarrow\infty$, $\tilde c_n(\theta_1^0)\big(-2\log \ELR_1(\theta_1^0)\big)\rightarrow \chi_1^2$  in distribution when $\theta_1^0>0$, and $\tilde c_n(0)(-2\log \ELR_1(0))\rightarrow U_+^2$ in distribution, where $U\sim N(0,1)$ and $U_+=\max(U,0)$.
\end{theorem}
\begin{proof}\label{pf:th:1}
    For simplicity, we sometimes use $Z_i$ to denote $Z_i(\theta_1)$ when there is no confusion.
    Using the method of Lagrange multipliers, let 
    \[
    \mathcal L = -\sum_{i=1}^n\log p_i + \kappa(\sum_{i=1}^np_i-1)+\lambda_0\sum_{i=1}^np_iZ_i.
    \]
    Since
    \[
    \frac{\partial \mathcal L}{\partial p_i} = -\frac{1}{p_i}+\kappa+\lambda_0 Z_i = 0,
    \]
    we have 
    \begin{equation}\label{eq:p1}
        p_i = \frac{1}{\kappa+\lambda_0 Z_i} \quad \text{ and } \quad \kappa = n.
    \end{equation}
    Plugging $\lambda_0=n\lambda$ into (\ref{eq:p1}), we obtain
    \begin{equation}\label{eq:p2}
        p_i = \frac{1}{n(1+\lambda Z_i)}.
    \end{equation}
    Since
    \begin{equation}\label{eq:Q1}
        0 = \sum_{i=1}^np_iZ_i = \sum_{i=1}^n\frac{Z_i}{n(1+\lambda Z_i)},
    \end{equation}
    under Condition \ref{cond:1-2}, one can show that
    \[
    \lambda = \big(\sum_{i=1}^nZ_i^2\big)^{-1}{\sum_{i=1}^n Z_i} + o_p(n^{-1/2}) \text{ by Taylor expansion}.
    \]
    Let $W_1(\theta_1)=n^nL_1(\theta_1)$. We have 
    \begin{align}\label{eq:R1}
        -2\log (W_1(\theta_1)) & = -2\sum_{i=1}^n(\log p_i+\log n) = 2\sum_{i=1}^n\log(1+\lambda Z_i) \notag\\
        & = 2\sum_{i=1}^n\big(\lambda Z_i-\frac{1}{2}(\lambda Z_i)^2\big)+o_p(1) \text{ (by Taylor expansion)} \notag\\
        & = \big(\sum_{i=1}^nZ_i\big)^2\big(\sum_{i=1}^nZ_i^2\big)^{-1}+o_p(1).
    \end{align}
    
    \begin{enumerate}
        \item[(1)] If $\theta_1^0>0$, then Lemma \ref{lem:theta} implies 
        \begin{align*}
            & \tilde c_n(\theta_1^0)\big(-2\log \frac{L_1(\theta_1^0)}{\max_{\theta_1\geq 0}L_1(\theta_1)}\big) = \tilde c_n(\theta_1^0)\big(-2\log W_1(\theta_1^0)\big) \\
            = & \big(n^{-1/2}\sum_{i=1}^nZ_i(\theta_1^0)\big)^2/\tilde\nu_{1n}^2(\theta_1^0) +o_p(1) \\
            = & \frac{\big( n^{-1/2}\alpha^{-1}\sum_{i=1}^n\big\langle\Phi_{i1}-\sum_{q=1}^{d-1}F_q\Phi_{iq+1}, R_i-\theta_1^0\Phi_{i1}\big\rangle\big)^2}{n^{-1}\alpha^{-2}\sum_{i=1}^n\big\langle R_i-H_i((\theta_1^0,\tilde{\theta}_{(1)}^T)^T), \Phi_{i1}-\sum_{q=1}^{d-1}F_q\Phi_{iq+1}\big\rangle^2} +o_p(1).
        \end{align*}
        Since
        \begin{align*}
            & n^{-1}\alpha^{-2}\sum_{i=1}^n\big\langle R_i-H_i((\theta_1^0,\tilde{\theta}_{(1)}^T)^T), \Phi_{i1}-\sum_{q=1}^{d-1}F_q\Phi_{iq+1}\big\rangle^2 \\
            = & \var\big(n^{-1/2}\alpha^{-1}\sum_{i=1}^n\big\langle\Phi_{i1}-\sum_{q=1}^{d-1}F_q\Phi_{iq+1}, R_i-\theta_1^0\Phi_{i1}\big\rangle\big) + o_p(1)
        \end{align*}
        under Condition \ref{cond:1-2},  it implies $\tilde c_n(\theta_1^0)\big(-2\log \ELR_1(\theta_1^0)\big)\rightarrow\chi^2_1$ in distribution when $\theta_1^0>0$.
        
        \item[(2)] When $\theta_1^0=0$, Lemma \ref{lem:theta} implies
        \[
        \tilde c_n(0)\big(-2\log \frac{L_1(0)}{\max_{\theta_1\geq 0}L_1(\theta_1)}\big) = 
        \tilde c_n(0)\big(-2\log W_1(0)\big) I\big(\sum_{i=1}^nZ_i(0)\geq 0\big) 
        \]
         as $n$ is large enough.
        Therefore, $\tilde c_n(0)\big(-2\log \ELR_1(0)\big) \rightarrow U_+^2$
        in distribution, where $U\sim N(0,1)$ and $U_+=\max(U,0)$.
    \end{enumerate}
\end{proof}

\begin{lemma}\label{lem:theta}
    Let $W_1(\theta_1)=n^nL_1(\theta_1)$ and $\tilde\theta_1 = \arg\max_{\theta_1\geq 0} W_1(\theta_1)$.
    If the true value $\theta^*_1>0$, then $W_1(\tilde\theta_1)=1$ as $n$ is large enough. If the true value $\theta^*_1=0$, then $W_1(\tilde\theta_1)=I(\sum_{i=1}^nZ_i(0)\geq 0)+W_1(0)I(\sum_{i=1}^nZ_i(0)< 0)$ as $n$ is large enough.  
\end{lemma}
\begin{proof}
    Let $\check\theta_1= \arg\max_{\theta_1} W_1(\theta_1)$.
    We use a similar method in \cite{qin1994empirical}.
    Since
    \[
    \check\theta_1= \arg\min_{\theta_1} -2\log W_1(\theta_1),
    \]
    \begin{align}\label{eq:Q2}
        0=\frac{\partial(-2\log W_1(\theta_1))}{\partial\theta_1}|_{\theta_1=\check\theta_1} & = 2\sum_{i=1}^n\frac{\frac{\partial\lambda}{\partial\theta_1}Z_i+\lambda\frac{\partial Z_i}{\partial\theta_1}}{1+\lambda Z_i}|_{\theta_1=\check\theta_1}\notag \\
        & = 2\lambda\sum_{i=1}^n\frac{1}{1+\lambda Z_i}\frac{\partial Z_i}{\partial \theta_1}|_{\theta_1=\check\theta_1} \text{ (by (\ref{eq:Q1}))}.
    \end{align}
    Let $\check\lambda = \lambda(\check\theta_1)$. We note $\check\theta_1$ and $\check\lambda$ satisfy
    \[
    Q_{1n}(\check\theta_1,\check\lambda)  = 0, \quad Q_{2n}(\check\theta_1,\check\lambda) =0,
    \]
    where
    \begin{align*}
        Q_{1n}(\theta_1,\lambda) & = \frac{1}{n}\sum_{i=1}^n\frac{Z_i(\theta_1)}{1+\lambda Z_i(\theta_1)} \text{ (by (\ref{eq:Q1}))},\\
        Q_{2n}(\theta_1,\lambda) & = \frac{\lambda}{n}\sum_{i=1}^n\frac{1}{1+\lambda Z_i(\theta_1)}\frac{\partial Z_i(\theta_1)}{\partial \theta_1} \text{ (by (\ref{eq:Q2}))}.
    \end{align*}
    Taking derivatives about $\theta_1$ and $\lambda$, we have
    \begin{align*}
        & \frac{\partial Q_{1n}(\theta_1,0)}{\partial\theta_1} =\frac{1}{n} \sum_{i=1}^n\frac{\partial Z_i(\theta_1)}{\partial\theta_1}, & \frac{\partial Q_{1n}(\theta_1,0)}{\partial\lambda} = -\frac{1}{n}\sum_{i=1}^nZ_i(\theta_1)^2, \\
        & \frac{\partial Q_{2n}(\theta_1,0)}{\partial\theta_1} = 0,& \frac{\partial Q_{2n}(\theta_1,0)}{\partial\lambda} = \frac{1}{n}\sum_{i=1}^n\frac{\partial Z_i(\theta_1)}{\partial\theta_1}.
    \end{align*}
    Expanding $Q_{1n}$ and $Q_{2n}$ at $(\theta_1=\theta^*_1,\lambda=0)$, we have 
    \begin{align}
        0 = & Q_{1n}(\check\theta_1,\check\lambda) \notag \\
        = & Q_{1n}(\theta^*_1,0) + \frac{\partial Q_{1n}(\theta^*_1,0)}{\partial\theta_1}(\check\theta_1-\theta^*_1) + \frac{\partial Q_{1n}(\theta^*_1,0)}{\partial\lambda}\check\lambda + o_p(n^{-1/2}),  \label{eq:Q10} \\
        0 = & Q_{2n}(\check\theta_1,\check\lambda) \notag \\
        = & Q_{2n}(\theta^*_1,0) + \frac{\partial Q_{2n}(\theta^*_1,0)}{\partial\theta_1}(\check\theta_1-\theta^*_1) + \frac{\partial Q_{2n}(\theta_1^*,0)}{\partial\lambda}\check\lambda + o_p(n^{-1/2}). \label{eq:Q20}
    \end{align}
    (\ref{eq:Q10}) and (\ref{eq:Q20}) give
    \[
    \begin{pmatrix}
        \check\theta_1 - \theta^*_1 \\
        \check\lambda
    \end{pmatrix}=\begin{pmatrix}
        \frac{1}{n} \sum_{i=1}^n\frac{\partial Z_i(\theta^*_1)}{\partial\theta_1} & -\frac{1}{n} \sum_{i=1}^nZ_i(\theta^*_1)^2 \\
        0 & \frac{1}{n} \sum_{i=1}^n\frac{\partial Z_i(\theta^*_1)}{\partial\theta_1}
    \end{pmatrix}^{-1}\begin{pmatrix}
        -\frac{1}{n}\sum_{i=1}^nZ_i(\theta^*_1)+o_p(n^{-1/2})\\
        o_p(n^{-1/2})
    \end{pmatrix}
    \]
    Hence,
    \begin{equation}\label{eq:thetacheck}
        \check\theta_1-\theta^*_1 = -\frac{n^{-1}\sum_{i=1}^nZ_i(\theta^*_1)}{n^{-1}\sum_{i=1}^n\frac{\partial Z_i(\theta^*_1)}{\partial\theta_1}}+o_p(n^{-1/2}),
    \end{equation}
    where $(\sum_{i=1}^nZ_i(\theta^*_1))/(\sum_{i=1}^n\partial Z_i(\theta^*_1)/\partial\theta_1)=O_p(n^{-1/2})$.
    
    When $\theta^*_1>0$, (\ref{eq:thetacheck}) implies $\check\theta_1>0$ as $n$ is large enough. Thus, $\tilde{\theta}_1=\check\theta_1$ as $n$ is large enough. 
    Then $\tilde{\theta}_1$ satisfies (\ref{eq:Q2}), i.e.,
    \begin{equation}\label{eq:lam}
        2\lambda\sum_{i=1}^n\frac{1}{1+\lambda Z_i}(-\|\Phi_{i1}\|_F^2) = 0.
    \end{equation}
    Plugging (\ref{eq:p2}) into (\ref{eq:lam}), we have $\lambda=0$. Then $p_i=n^{-1}$ and $W_1(\tilde\theta_1)=1$.
    
    When $\theta^*_1=0$, $\tilde{\theta}_1=\check\theta_1I(\check\theta_1\geq 0)$ as $n$ is large enough.
    Since $\sum_{i=1}^n\partial Z_i(0)/\partial\theta_1=-\sum_{i=1}^n\|\Phi_{i1}\|_F^2<0$, we have $\tilde\theta_1=\check\theta_1I(\sum_{i=1}^nZ_i(0)\geq 0)$ as $n$ is large enough. So $W_1(\tilde\theta_1)=I(\sum_{i=1}^nZ_i(0)\geq 0)+W_1(0)I(\sum_{i=1}^nZ_i(0)< 0)$.
\end{proof}

\begin{proof}[of Theorem \ref{th:3}]
    Let $\Delta$ be a $d$-dimensional vector with the $k$th element $\Delta_k=\sum_{i=1}^n\tr(\Phi_{ik}E(\hat\epsilon_i))$. 
    Let $\varsigma_i=(\tr(\Phi_{i1}\Phi_{i2}),\cdots,\tr(\Phi_{i1}\Phi_{id}))^T$.
    For $i=1,\cdots,n$,
    \begin{align*}
        & E(\hat R_i) = H_i(\theta^*)+E(\hat\epsilon_i), \\
        & E(\hat\theta_{(1)}) = \theta^*_{(1)}+\big(\Xi^{-1}\big)_{-1}^T\Delta,
    \end{align*}
    so we have
    \begin{align}\label{eq:hatZi}
        & E(\hat Z_i(\theta_1^0)) = \tr(\Phi_{i1}E(\hat\epsilon_i))-\varsigma_i^T(\Xi^{-1})_{-1}^T\Delta=O(n^{-1}), \\
        & E(n^{-1/2}\sum_{i=1}^n\hat Z_i(\theta_1^0)) = o(1) \notag
    \end{align}
    under Proposition \ref{prop:1}.
    Then with similar techniques in the proof of Theorem \ref{th:1}, Theorem \ref{th:3} can be proved.
\end{proof}

\subsection{Proof of of Proposition \ref{prop:1}}
\begin{proof}[of Proposition \ref{prop:1}]
    Since 
    \[
    n^{1/2}(\hat{\beta}-\beta^*) = n^{1/2}(X^TX)^{-1}X^Tr = (n^{-1}X^TX)^{-1}n^{-1/2}X^Tr,
    \]
    we have $n^{1/2}(\hat\beta-\beta^*)\xrightarrow{d}\Sigma^{-1}\eta$.
    
    Under Condition \ref{cond:1-2}, $E\|r_i\|_2^2$ and $E\|r_i\|_2^4$ are bounded uniformly.
    Since 
    \begin{align*}
        nE(r_i(\beta^*-\hat\beta)^TX_i^T)=& -E(r_i\sum_{k=1}^nr_k^TX_k(n^{-1}X^TX)^{-1}X_i^T) \\
        = & -E(r_ir_i^TX_i(n^{-1}X^TX)^{-1} X_i^T)\rightarrow O(1),
    \end{align*}
    \[
    nE(X_i(\beta^*-\hat\beta)(\beta^*-\hat\beta)^TX_i^T)\rightarrow O(1),
    \]
    we have $E({\hat\epsilon_i})=O(n^{-1})$.
    
    Note that for $i\neq j$,
    \begin{align*}
        & \cov({r_ir_i^T},{\hat\epsilon_j}) \\
        = & \cov({r_ir_i^T},{r_j(\beta^*-\hat{\beta})^TX_j^T}) + \cov({r_ir_i^T},{X_j(\beta^*-\hat{\beta})r_j^T}) + \cov({r_ir_i^T},{X_j(\beta^*-\hat{\beta})(\beta^*-\hat{\beta})^TX_j^T}) \\
        = & \cov({r_ir_i^T},{X_j(\beta^*-\hat{\beta})(\beta^*-\hat{\beta})^TX_j^T}).
    \end{align*}
    Since
    \begin{align*}
        & n^2\cov({r_ir_i^T},{X_j(\beta^*-\hat{\beta})(\beta^*-\hat{\beta})^TX_j^T}) \\
        = & \cov({r_ir_i^T},{X_j(n^{-1}X^TX)^{-1}\sum_{l=1}^nX_l^Tr_l\sum_{k=1}^nr_k^TX_k(n^{-1}X^TX)^{-1}X_j^T}) \\
        \rightarrow & \cov({r_ir_i^T},{X_j\Sigma^{-1}X_i^Tr_ir_i^TX_i\Sigma^{-1}X_j^T}) = O(1),
    \end{align*}
    $\cov({r_ir_i^T},{\hat\epsilon_j})=O(n^{-2})$.
    
    For $\cov({\hat\epsilon_i},{\hat\epsilon_j}),~i\neq j$, we only analyze $\cov({r_i(\beta^*-\hat{\beta})^TX_i^T},{r_j(\beta^*-\hat{\beta})^TX_j^T})$, $\cov({r_i(\beta^*-\hat{\beta})^TX_i^T},$ ${X_j(\beta^*-\hat{\beta})(\beta^*-\hat{\beta})^TX_j^T})$ and $\cov({X_i(\beta^*-\hat{\beta})(\beta^*-\hat{\beta})^TX_i^T},$ ${X_j(\beta^*-\hat{\beta})(\beta^*-\hat{\beta})^TX_j^T})$.
    Since
    \begin{align*}
        & n^2\cov({r_i(\beta^*-\hat{\beta})^TX_i^T},{r_j(\beta^*-\hat{\beta})^TX_j^T}) \\
        = & \cov({r_i\sum_{l=1}^nr_l^TX_l(n^{-1}X^TX)^{-1}X_i^T},{r_j\sum_{k=1}^nr_k^TX_k(n^{-1}X^TX)^{-1}X_j^T}) \\
        \rightarrow & \cov({r_ir_j^TX_j\Sigma^{-1}X_i^T},{r_jr_i^TX_i\Sigma^{-1}X_j^T})=O(1),
    \end{align*}
    \[
    n^2\cov({X_i(\beta^*-\hat{\beta})(\beta^*-\hat{\beta})^TX_i^T},{X_j(\beta^*-\hat{\beta})(\beta^*-\hat{\beta})^TX_j^T})\rightarrow O(1),
    \]
    \begin{align*}
        & \cov({r_i(\beta^*-\hat{\beta})^TX_i^T},{X_j(\beta^*-\hat{\beta})(\beta^*-\hat{\beta})^TX_j^T}) \\
        = & -n^{-3}\cov({r_i\sum_{k=1}^nr_k^TX_k(n^{-1}X^TX)^{-1}X_i^T},{X_j(n^{-1}X^TX)^{-1}\sum_{s=1}^n\sum_{t=1}^nX_s^Tr_sr_t^TX_t(n^{-1}X^TX)^{-1}X_j^T}) \\
        = & -n^{-3}\cov({r_ir_i^TX_i(n^{-1}X^TX)^{-1}X_i^T},{X_j(n^{-1}X^TX)^{-1}X_i^Tr_ir_i^TX_i(n^{-1}X^TX)^{-1}X_j^T}) \\
        &  -n^{-3}\sum_{k\neq i}\cov({r_ir_k^TX_k(n^{-1}X^TX)^{-1}X_i^T},{X_j(n^{-1}X^TX)^{-1}(X_i^Tr_ir_k^TX_k+X_k^Tr_kr_i^TX_i)(n^{-1}X^TX)^{-1}X_j^T}), 
    \end{align*}
    we have $\cov({r_i(\beta^*-\hat{\beta})^TX_i^T},{r_j(\beta^*-\hat{\beta})^TX_j^T})$, $\cov({r_i(\beta^*-\hat{\beta})^TX_i^T},{X_j(\beta^*-\hat{\beta})(\beta^*-\hat{\beta})^TX_j^T})$, and $\cov({X_i(\beta^*-\hat{\beta})(\beta^*-\hat{\beta})^TX_i^T},{X_j(\beta^*-\hat{\beta})(\beta^*-\hat{\beta})^TX_j^T})=O(n^{-2})$.
    Hence, $\cov({\hat\epsilon_i},{\hat\epsilon_j})=O(n^{-2})$.
\end{proof}

\subsection{proof of equation \eqref{eq:zdm}}
\begin{proof}[of equation \eqref{eq:zdm}]
    Rewrite $\Xi$ as $\Xi=\begin{psmallmatrix} E_{11} & E_{12} \\ E_{21} & E_{22}
    \end{psmallmatrix}$ with $E_{11}$ being a scalar. Rewrite $\hat\Upsilon$ as $\hat\Upsilon=(\hat\Upsilon_1,\hat\Upsilon_{(1)})^T$.
    So
    \begin{align*}
        \hat{\theta}_{(1)} = (\Xi^{-1})_{-1}^T\hat\Upsilon = -E_{22}^{-1}E_{21}q^{-1}\hat\Upsilon_1+E_{22}^{-1}\hat\Upsilon_{(1)}+E_{22}^{-1}E_{21}q^{-1}E_{12}E_{22}^{-1}\hat\Upsilon_{(1)},
    \end{align*}
    where $q=E_{11}-E_{12}E_{22}^{-1}E_{21}$.
    Let $F=E_{22}^{-1}E_{21}$. We obtain
    \begin{align*}
        \sum_{i=1}^n\hat Z_i(\theta_1^0) = &\hat\Upsilon_1 - F^T\big(-q^{-1}E_{21}\hat\Upsilon_1+\hat\Upsilon_{(1)}+E_{21}q^{-1}E_{12}E_{22}^{-1}\hat\Upsilon_{(1)}\big) -E_{11}\theta_1^0 \\
        = & (1+q^{-1}F^TE_{21})\hat\Upsilon_1 - (1+F^TE_{21}q^{-1})F^T\hat\Upsilon_{(1)} -E_{11}\theta_1^0 \\
        = & (1+q^{-1}F^TE_{21})\sum_{i=1}^n\big\langle\Phi_{i1}-\sum_{q=1}^{d-1}F_q\Phi_{iq+1},\hat R_i\big\rangle -E_{11}\theta_1^0 \\
        = & \sum_{i=1}^n\hat D_i(\theta_1^0),
    \end{align*}
    where $\hat D_i(\theta_1^0)=\alpha^{-1}\big\langle\Phi_{i1}-\sum_{q=1}^{d-1}F_q\Phi_{iq+1},\hat R_i-\theta_1^0\Phi_{i1}\big\rangle$.
    Note that for any $b=(b_1,\cdots,b_{d-1})^T$,
    \begin{align*}
        \sum_{i=1}^n\big\langle\Phi_{i1}-\sum_{q=1}^{d-1}F_q\Phi_{iq+1},\sum_{j=1}^{d-1}b_j\Phi_{ij+1}\big\rangle = \sum_{j=1}^{d-1}b_j\Xi_{1j+1}-\sum_{j=1}^{d-1}b_j\sum_{q=1}^{d-1}F_q\Xi_{q+1j+1} = 0,
    \end{align*}
    so we have 
    \[
    \sum_{i=1}^n\hat D_i(\theta_1^0)=\sum_{i=1}^n\hat M_i(\theta_1^0),
    \]
    where $\hat M_i(\theta_1^0)=\alpha^{-1}\langle\Phi_{i1}-\sum_{q=1}^{d-1}F_q\Phi_{iq+1},\hat R_i-H_i((\theta_1^0,\hat\theta_{(1)}^T)^T)\rangle$.
\end{proof}

\subsection{Proof of Proposition \ref{prop:2}}
\begin{proof}[of Proposition \ref{prop:2}]
    Since $\xi_i^{(g)}\sim N(0,1)$ and \eqref{eq:hatZi}, property (i) holds.
    
    To prove property (ii), we have 
    \begin{align*}
        \var\big(n^{-1/2}\sum_{i=1}^n\hat M_i(\theta_1^0,t)\xi_i^{(g)}\big)=n^{-1}\sum_{i=1}^n\var(\hat M_i(\theta_1^0,t)\xi_i^{(g)})=n^{-1}\sum_{i=1}^nE\hat M_i(\theta_1^0,t)^2,
    \end{align*}
    and 
    \begin{align*}
        \var\big(n^{-1/2}\sum_{i=1}^n\hat Z_i(\theta_1^0,t)\big) = & \var\big(n^{-1/2}\sum_{i=1}^n\hat D_i(\theta_1^0,t)\big) 
        =  n^{-1}\sum_{i=1}^n\var(\hat D_i(\theta_1^0,t)) + o(1) \\
        = & n^{-1}\sum_{i=1}^nE\hat M_i(\theta_1^0,t)^2 + o(1).
    \end{align*}
    Thus, we have property (ii).
    
    Since 
    \begin{align*}
        & \cov\big(n^{-1/2}\sum_{i=1}^n\hat M_i(\theta_1^0,s)\xi_i^{(g)},n^{-1/2}\sum_{j=1}^n\hat M_j(\theta_1^0,t)\xi_j^{(g)}\big) 
        = n^{-1}\sum_{i=1}^n\cov(\hat M_i(\theta_1^0,s)\xi_i^{(g)},\hat M_i(\theta_1^0,t)\xi_i^{(g)}) \\
        = & n^{-1}\sum_{i=1}^nE(\hat M_i(\theta_1^0,s)\hat M_i(\theta_1^0,t)),
    \end{align*}
    and 
    \begin{align*}
        & \cov\big(n^{-1/2}\sum_{i=1}^n\hat Z_i(\theta_1^0,s),n^{-1/2}\sum_{j=1}^n\hat Z_j(\theta_1^0,t)\big) 
        =  \cov\big(n^{-1/2}\sum_{i=1}^n\hat D_i(\theta_1^0,s),n^{-1/2}\sum_{j=1}^n\hat D_j(\theta_1^0,t)\big) \\
        = & n^{-1}\sum_{i=1}^n\cov(\hat D_i(\theta_1^0,s),\hat D_i(\theta_1^0,t))+o(1) 
        =  n^{-1}\sum_{i=1}^nE(\hat M_i(\theta_1^0,s)\hat M_i(\theta_1^0,t))+o(1),
    \end{align*}
    we see property (iii) holds.
\end{proof}

\bibliographystyle{apalike}
\bibliography{eltest}

\end{document}